\documentclass[sts]{imsart}

\RequirePackage{amsthm,amsmath,amsfonts,amssymb}
\RequirePackage[numbers]{natbib}
\RequirePackage[colorlinks,citecolor=blue,urlcolor=blue]{hyperref}
\RequirePackage{graphicx}
\usepackage{algorithm}
\usepackage{algorithmic}
\usepackage{booktabs}
\usepackage{bbm}
\usepackage{ulem}
\usepackage{cancel}

\startlocaldefs
\theoremstyle{plain}

\newtheorem{theorem}{Theorem}[section]

\theoremstyle{definition}

\newtheorem{proposition}[theorem]{Proposition}
\endlocaldefs

\def\vec{\mathop{\rm vec}\nolimits}

\def\dist{\mathop{\rm dist}\nolimits}

\def\prox{\mathop{\rm prox}\nolimits}
\def\argmin{\mathop{\rm argmin}\nolimits}

\newcommand{\bzero}{\boldsymbol{0}}
\newcommand{\bone}{\boldsymbol{1}}

\newcommand{\bb}{\boldsymbol{b}}

\newcommand{\bv}{\boldsymbol{v}}
\newcommand{\bw}{\boldsymbol{w}}
\newcommand{\bx}{\boldsymbol{x}}
\newcommand{\by}{\boldsymbol{y}}
\newcommand{\bz}{\boldsymbol{z}}
\newcommand{\bA}{\boldsymbol{A}}
\newcommand{\bB}{\boldsymbol{B}}
\newcommand{\bC}{\boldsymbol{C}}
\newcommand{\bD}{\boldsymbol{D}}
\newcommand{\bE}{\boldsymbol{E}}

\newcommand{\bH}{\boldsymbol{H}}
\newcommand{\bI}{\boldsymbol{I}}

\newcommand{\bL}{\boldsymbol{L}}

\newcommand{\bQ}{\boldsymbol{Q}}

\newcommand{\bU}{\boldsymbol{U}}
\newcommand{\bV}{\boldsymbol{V}}
\newcommand{\bW}{\boldsymbol{W}}
\newcommand{\bX}{\boldsymbol{X}}
\newcommand{\bY}{\boldsymbol{Y}}
\newcommand{\bZ}{\boldsymbol{Z}}
\newcommand{\balpha}{\boldsymbol{\alpha}}
\newcommand{\bbeta}{\boldsymbol{\beta}}

\newcommand{\bgamma}{\boldsymbol{\gamma}}
\newcommand{\bdelta}{\boldsymbol{\delta}}

\newcommand{\bmu}{\boldsymbol{\mu}}
\newcommand{\bnu}{\boldsymbol{\nu}}

\newcommand{\bsigma}{\boldsymbol{\sigma}}

\newcommand{\bDelta}{\boldsymbol{\Delta}}

\newcommand{\bLambda}{\boldsymbol{\Lambda}}

\newcommand{\bSigma}{\boldsymbol{\Sigma}}

\newcommand{\bPhi}{\boldsymbol{\Phi}}

\newcommand{\Real}{\mathbb{R}}
\newcommand{\rank}{\mathrm{rank}}

\begin{document}

\begin{frontmatter}
\title{Tactics for Improving Least Squares
Estimation}
\runtitle{Deweighting Weighted Least Squares with Majorization Minimization}

\begin{aug}
\author[A]{\fnms{Qiang}~\snm{Heng}\ead[label=e1]{qheng.stat@seu.edu.cn}\orcid{0000-0002-4042-6773}}
\author[B]{\fnms{Hua}~\snm{Zhou}\ead[label=e2]{huazhou@ucla.edu}\orcid{0000-0003-1320-7118}}
\author[C]{\fnms{Kenneth}~\snm{Lange}\ead[label=e3]{klange@ucla.edu}\orcid{0000-0002-1313-5030}}


\address[A]{Qiang Heng is Associate Researcher, School of Statistics and Data Science, Southeast University, Nanjing, China. The majority of this work was completed during his postdoc training at UCLA.\printead[presep={\ }]{e1}.}

\address[B]{Hua Zhou is Professor, Departments of Biostatistics and Computational Medicine,
University of California, Los Angeles, USA\printead[presep={\ }]{e2}.}

\address[C]{Kenneth Lange is Professor, Departments of Computational Medicine, Human Genetics, and Statistics, University of California, Los Angeles, USA\printead[presep={\ }]{e3}.}

\end{aug}

\begin{abstract}
This paper deals with tactics for fast computation in least squares regression in high dimensions. These tactics include: (a) the majorization-minimization (MM) principle, (b) smoothing by Moreau envelopes, and (c) the proximal distance principle for constrained estimation. In iteratively reweighted least squares, the MM principle can create a surrogate function that trades case weights for adjusted responses. Reduction to ordinary least squares then permits the reuse of the Gram matrix and its Cholesky decomposition across iterations. This tactic is pertinent to estimation in L2E regression and generalized linear models. For problems such as quantile regression, non-smooth terms of an objective function can be replaced by their Moreau envelope approximations and majorized by spherical quadratics. Finally, penalized regression with distance-to-set penalties also benefits from this perspective.  Our numerical experiments validate the speed and utility of deweighting and Moreau envelope approximations. Julia software implementing these experiments is available on our web page.
\end{abstract}


\end{frontmatter}

\section{Introduction}

It is fair to say that most estimation problems in classical statistics reduce to either least squares or maximum likelihood and that much of maximum likelihood estimation reduces to iteratively reweighted least squares (IRLS). Quantile regression and generalized linear model (GLM) regression are just two of the many examples of IRLS. The modern era of frequentist statistics has added to the classical inference soup sparsity, robustness, and rank restrictions. These additions complicate estimation by introducing nonsmooth terms into the objective function. In practice, smooth approximations can be substituted for nonsmooth terms without substantially impacting estimates or inference. The present paper documents the beneficial effects on model selection and solution speed of several tactics: (a) systematic application of the majorization-minimization (MM) principle, (b) substitution of ordinary least squares for weighted least squares, (c) substitution of Moreau envelopes for nonsmooth losses and penalties, and (d) application of the proximal distance principle in constrained estimation. 

The speed of an optimization algorithm is a trade-off between the cost per iteration and the number of iterations until convergence. In our view, this trade-off has not been adequately explored in IRLS.  Recall that in IRLS, we minimize the objective
\begin{eqnarray}
f(\bbeta) & = & \frac{1}{2}\sum_{i=1}^n w_{mi}(y_i-\bx_i^\top\bbeta)^2,
\label{IRLS_criterion}
\end{eqnarray}
at iteration $m$. Here $\by=(y_i)$ is the response vector, $\bX = (x_{ij})$ is the design matrix, and $\bbeta=(\beta_j)$ is the vector of regression coefficients. The dominant cost in high-dimensional IRLS is the solution of the normal equations $(\bX^\top \bW_m \bX) \bbeta = \bX^\top \bW_m \by$, where $\bW_m$ is the diagonal matrix of case weights. For dense problems, solution of the normal equations can be achieved via the Cholesky decomposition of the weighted Gram matrix $\bX^\top \bW_m\bX$ or the QR decomposition of the weighted design matrix $\bW_m^{1/2}\bX$.

Unfortunately, both decompositions must be recalculated from scratch each time the weight matrix $\bW_m$ changes. Given $p\le n$ regression coefficients, the complexity of computing the Gram matrix, weighted or unweighted, is $O(np^2)$. Fortunately, matrix multiplication is super fast on modern computers that have access to highly optimized BLAS routines. On large-scale problems with $p$ comparable to $n$, Cholesky and QR decompositions execute more slowly than matrix multiplications even though the former has computational complexity $O(p^3)$, and the latter has computational complexity $O(np^2)$. If it takes $m$ iterations until convergence, these costs must be multiplied by $m$.  In repeated unweighted regressions, a single QR decomposition or a single paired Gram matrix and Cholesky decomposition suffices, and the computational burden drops to $O(p^3)$. On the other hand, the process of deweighting may well cause  $m$ to increase.  The downside of the Cholesky decomposition approach is that the condition number of $\bX^\top\bX$ is the square of the condition number of $\bX$. In our view, this danger is overrated. Adding a small ridge penalty overcomes it and is still consistent with rapid extraction of the Cholesky decomposition. In practice, code based on Cholesky decomposition is faster than code based on QR decomposition, and our numerical experiments exploit this fact.

The majorization-minimization (MM) principle \cite{hunter2004atutorial,lange2016mm,lange2000optimization} is our engine for leveraging Moreau envelopes and converting weighted to unweighted least squares. In minimization, MM iteratively substitutes a surrogate function $g(\bbeta \mid \bbeta_m)$ that majorizes a loss $f(\bbeta)$ around the current iterate $\bbeta_m$. Majorization is defined by the tangency condition $g(\bbeta_m \mid \bbeta_m) =f(\bbeta_m)$ and the domination condition $g(\bbeta \mid \bbeta_m)\geq f(\bbeta)$ for all $\bbeta$.  The surrogate balances two goals: hugging the objective tightly and simplifying minimization. Minimizing the surrogate produces the next iterate $\bbeta_{m+1}$ and drives the objective downhill owing to the conditions
\begin{eqnarray*}
f(\bbeta_{m+1}) & \le & g(\bbeta_{m+1} \mid \bbeta_m)  \le  g(\bbeta_m \mid \bbeta_m)  =  f(\bbeta_m).
\end{eqnarray*}
In maximization, the surrogate minorizes the objective and must be maximized instead. The tangency condition remains intact, but the domination condition $g(\bbeta \mid \bbeta_m)\leq f(\bbeta)$ is now reversed. The celebrated EM (expectation-maximization) principle for maximum likelihood estimation with missing data \cite{mclachlan2007algorithm} is a special case of minorization-maximization. In the EM setting, Jensen's inequality supplies the surrogate as the expectation of the complete data log-likelihood conditional on the observed data.

In practice, many surrogate functions are strictly convex quadratics. When this is the case,  minimizing the surrogate is achieved by the Newton update
\begin{eqnarray}
\bbeta_{m+1} & = & \bbeta_m-d^2 g(\bbeta_{m} \mid \bbeta_m)^{-1}
\nabla g(\bbeta_m \mid \bbeta_m) \nonumber \\
& = & \bbeta_m-d^2 g(\bbeta_{m} \mid \bbeta_m)^{-1}
\nabla f(\bbeta_m). \label{MMquadratic_update}
\end{eqnarray}
The second form of the update reflects the tangency condition $\nabla g(\bbeta_m \mid \bbeta_m)=\nabla f(\bbeta_m)$. In this version of gradient descent, the step length is exactly $1$. The local rate of convergence is determined by how well the surrogate curvature matrix $d^2g(\bbeta_{m} \mid \bbeta_m)$ approximates the actual curvature matrix $d^2f(\bbeta_m)$.  In the current paper, we focus on quadratic surrogates of fixed curvature, denoted by $\bH$. This approach has precedent in the early statistical literature. Examples include multinomial logistic regression \cite{bohning1992multinomial}, weighted least squares \cite{kiers1997weighted}, and support vector machines \cite{groenen2008svm}. These applications emphasize the virtue of inverting the curvature matrix just once. The present work builds on this tradition in several ways. We underscore the computational value of recycling matrix decompositions across iterations, offer new theoretical insights into the convergence of the quadratic MM algorithm, and demonstrate its utility across a diverse set of new examples.

The Moreau envelope of a function $f(\bbeta)$ from $\mathbb{R}^p$ to $\mathbb{R} \cup \{\infty\}$ is defined by
\begin{eqnarray*}
M_{\mu f}(\bbeta) & =& \inf_{\bnu}\Big[f(\bnu)+\frac{1}{2\mu}\|\bbeta-\bnu\|_2^2\Big].
\end{eqnarray*}
When $f(\bbeta)$ is convex, the infimum is attained, and its Moreau envelope $M_{\mu f}(\bbeta)$ is convex and continuously differentiable \cite{nesterov2005smooth}.  However, $M_{\mu f}(\bbeta)$ also exists for many nonconvex functions $f(\bbeta)$ \cite{proximal}.  In the absence of convexity, the proximal operator
\begin{eqnarray}
\prox_{\mu f}(\bbeta) & =& \underset{\bnu}{\argmin}\Big[f(\bnu)+\frac{1}{2\mu}\|\bbeta-\bnu\|_2^2\Big] \label{prox_map_eqn}
\end{eqnarray}
can map to a set rather than a single point. For a proxable function $f(\bbeta)$, the definition of the Moreau envelope implies the quadratic majorization
\begin{eqnarray}
M_{\mu f}(\bbeta) & \le & f(\bnu_m)+\frac{1}{2\mu}\|\bbeta-\bnu_m\|_2^2
\label{Moreau_majorization}
\end{eqnarray}
at $\bnu_m$ for any $\bnu_m \in \prox_{\mu f}(\bbeta_m)$. Generally, $M_{\mu f}(\bbeta) \le f(\bbeta)$ and $\lim_{\mu \downarrow 0} M_{\mu f}(\bbeta) = f(\bbeta)$ under mild conditions. In particular, if $f(\bbeta)$ is Lipschitz with constant L, then 
\begin{eqnarray*}
0 & \le & f(\bbeta) - M_{\mu f}(\bbeta) \le \frac{L^2 \mu}{2}.
\end{eqnarray*} 
Despite the elementary nature of the Moreau majorization (\ref{Moreau_majorization}),  its value in regression has largely gone unnoticed. 

In Section~\ref{conversion_section} we rederive the majorization of \cite{heiser1987correspondence}, replacing the weighted least squares criterion \eqref{IRLS_criterion} with the unweighted surrogate
\begin{equation}\label{deweighted_loss}
g(\bbeta\mid \bbeta_m) = \frac{1}{2} \sum_{i=1}^n 
\bigl[\tilde{y}_{mi} - \mu_i(\bbeta)\bigr]^2 + c_m,
\end{equation}
where $\tilde{y}_{mi} = w_{mi}y_i + (1-w_{mi})\mu_{mi}$ is the shifted response and $c_m$ is irrelevant in the subsequent minimization. In our examples, the regression functions $\mu_i(\bbeta)=\bx_i^\top\bbeta$ depend linearly on the parameter vector $\bbeta$. The majorization \eqref{deweighted_loss} removes the weights and shifts the responses, allowing the Cholesky decomposition of $\bX^\top\bX$ to be recycled across iterations.

\begin{figure}[htbp]
  \centering
  \includegraphics[width=\linewidth, height = 0.6\linewidth]{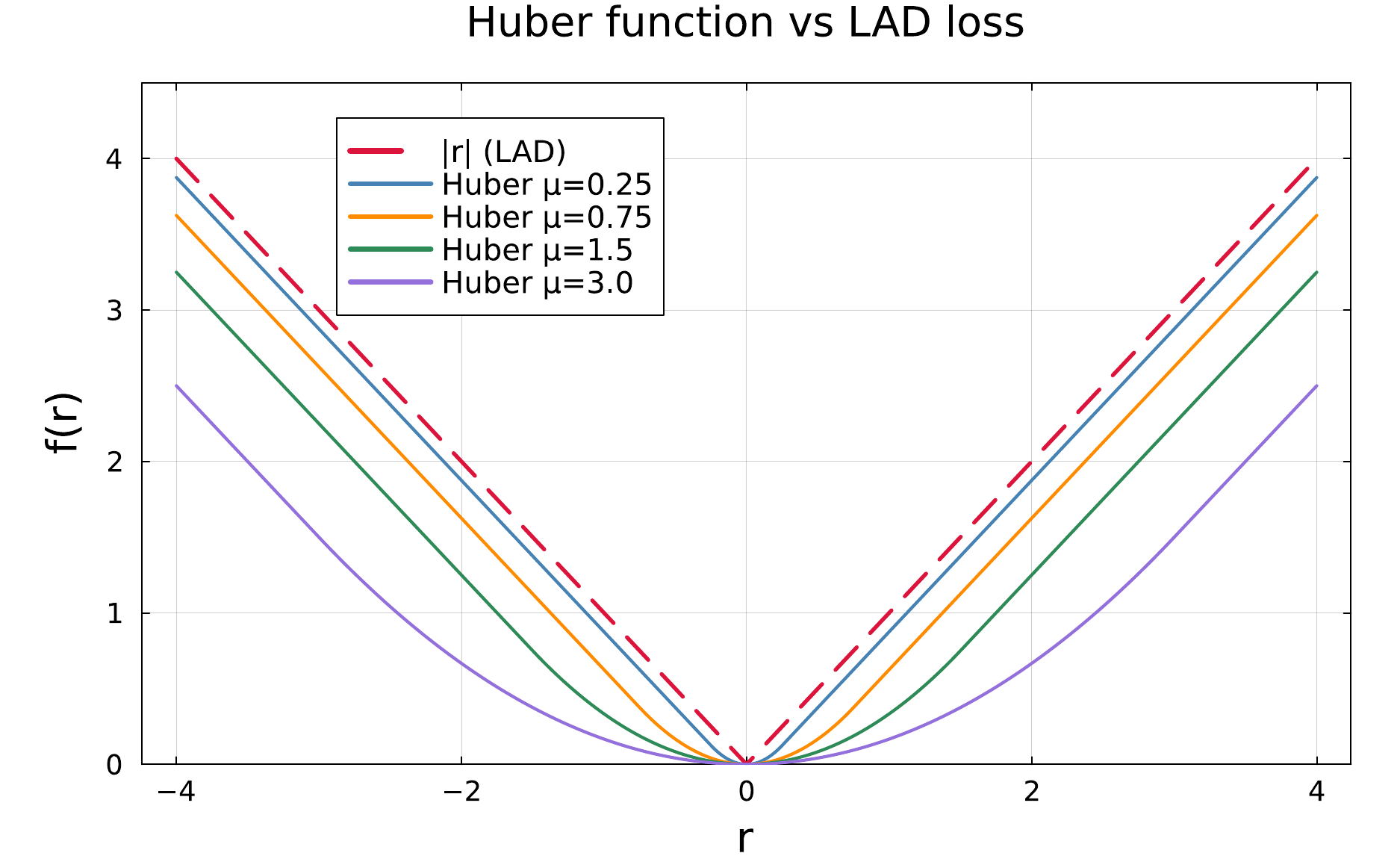}
  \caption{  Visualization of the Huber function with varying smoothing parameter $\mu$. }
  \label{huber}
\end{figure}

To render this discussion more concrete, consider the problem of least absolute deviation (LAD) regression with objective
\begin{eqnarray*}
f(\bbeta) & = & \sum_{i=1}^n |y_i-\bx_i^\top\bbeta|.
\end{eqnarray*}
The nonsmooth absolute value function has Moreau envelope equal to Huber's function
\begin{eqnarray*}
M_{\mu |\cdot|}(r) & = & \begin{cases} \frac{1}{2\mu}r^2, & |r| \le \mu  \\ |r|-\frac{\mu}{2}, & |r| > \mu 
\end{cases}
\end{eqnarray*}
with proximal map 
\begin{eqnarray}
\prox_{\mu |\cdot |}(r) & =&  \Big(1-\frac{\mu}{\max\{|r|,\mu\}}\Big)r. \label{LAD_prox}
\end{eqnarray}
The loose majorization (\ref{Moreau_majorization}) generates the least squares surrogate
\begin{eqnarray}\label{ladsurrogate}
g(\bbeta \mid \bbeta_m) & = & \frac{1}{2 \mu}\sum_{i=1}^n (r_i - z_{mi})^2,
\end{eqnarray}
where $r_i = y_i -\bx_i^\top\bbeta$, $r_{mi} = y_i -\bx_i^\top\bbeta_m$, and $z_{mi} = \prox_{\mu |\cdot|}(r_{mi})$.  Minimizing the surrogate reuses the same Cholesky decomposition at every iteration.

Alternatively, one can employ the best quadratic majorization \cite{de2009sharp}
\begin{eqnarray}\label{bestquadratic}
\quad M_{\mu |\cdot|}(r) & \le & 
\begin{cases}
\frac{1}{2|r_m|}(r-r_m)^{2}-r-\frac{1}{2}\mu & r_m \leq -\mu \\
\frac{1}{2\mu}r^{2}& |r_m|< \mu \\
\frac{1}{2|r_m|}(r-r_m)^{2}+r-\frac{1}{2}\mu & r_m \ge \mu \ 
\end{cases}.
\end{eqnarray}
This majorization translates into a weighted sum of squares plus an irrelevant constant. The resulting weighted sum of squares can be further deweighted by appealing to the surrogate \eqref{deweighted_loss}. 
Appendix \ref{double_major} shows that these two successive majorizations lead exactly to the Moreau envelope surrogate \eqref{ladsurrogate}.

Algorithm \ref{alg:LAD} summarizes a simple but effective algorithm for (smoothed) LAD regression that recycles the Cholesky decomposition of $\bX^\top\bX$.  Note that the smoothing parameter $\mu$ controls the quality of the approximation to the LAD loss; the smaller $\mu$, the more accurate the approximation. See Figure \ref{huber} for an illustration. However, as $\mu$ approaches 0, the Lipschitz constant of $M_{\mu|\cdot|}(r)$ grows, and the MM algorithm makes diminishing progress per iteration. A reasonable strategy is therefore to begin with a large value of $\mu$ and anneal it gradually toward zero. This acceleration strategy is further discussed in Section \ref{acceleration}. We also examine the impact of the smoothing parameter on the error of our estimators in Section \ref{impact}. 
\begin{algorithm}[tbp]
\caption{Fast LAD Regression}\label{alg:LAD}
\begin{algorithmic}[1]
\REQUIRE Design matrix $\bX\in \Real^{n\times p}$, response vector $\by \in \Real^{n}$,  smoothing constant $\mu>0$, and iteration number $m=0$.
\STATE Compute the Cholesky decomposition $\bL$ of $\bX^\top\bX$.
\STATE Initialize the regression coefficients $\bbeta_0= (\bL^\top)^{-1}(\bL^{-1}\bX^\top \by)$  by least squares.
\WHILE{not converged}
    \STATE For each case $i$ set $z_{mi}=\prox_{\mu |\cdot|}(y_i-\bx_i^\top \bbeta_m)$ based on formula   \eqref{LAD_prox}.
    \STATE $\bbeta_{m+1} \leftarrow (\bL^\top)^{-1}(\bL^{-1}\bX^\top (\by-\bz_m))$ update.
    \STATE $m \leftarrow m+1$ update.
\ENDWHILE
\ENSURE The final iterate $\bbeta_m$.
\end{algorithmic}
\end{algorithm}

Although the primary purpose of this paper is expository, we would like to highlight several likely new contributions. The first is our emphasis on the Moreau envelope majorization noted in equations \eqref{prox_map_eqn} and \eqref{Moreau_majorization} and recognized implicitly in the convex case by \cite{chen2012smoothing} and explicitly in general in the book \cite{lange2016mm}.  Among its many virtues, this majorization converts smoothed quantile regression into ordinary least squares.   The Moreau envelopes of various functions, such as the $\ell_0$-norm, the matrix rank function, and various set indicator functions, are invaluable in inducing sparse and low-rank structures.  The spherical majorization of the convolution-smoothed check function in Section \ref{quantile_regression_section1} and the subsequent MM algorithm is new to our knowledge. Our emphasis on the deweighting majorization of \cite{heiser1987correspondence} and \cite{kiers1997weighted} may help revive this important tactic. We also revisit the quadratic upper bound majorization of \cite{bohning1988monotonicity}. Application of this principle has long been a staple of logistic and multinomial regression \cite{bohning1992multinomial}. Inspired by the paper \cite{xu2022proximal}, we demonstrate how to reduce the normal equation of penalized multinomial regression to a Sylvester equation \cite{bartels1972algorithm}. Finally, in addition to mentioning existing theory guaranteeing the convergence of the algorithms studied here, we also present a new convergence proof based on fixed points and monotone operators. 

As a takeaway message from this paper, we hope readers will better appreciate the potential in high-dimensional estimation of combining the MM principle, smoothing, and numerical tactics such as recycling matrix decompositions, restarted Nesterov acceleration,  step-doubling, and smoothing constant annealing. These advances are ultimately as important as faster hardware.  Together the best hardware and the best algorithms will make future statistical analysis even more magical.

\section{Methods}
In this section, we first briefly review majorizations pertinent to least squares. The material covered here is largely standard. Good references are \cite{bauschke2011convex} and  \cite{beck2017first}.  After a few preliminaries, we take up: (a) conversion of weighted to unweighted least squares, (b) the quadratic bound principle, (c) integration with the proximal distance method, (d)  acceleration strategies, and (e) convergence theory.

\subsection{Notation}

Here are the notational conventions used throughout this article. All vectors and matrices appear in boldface. An entry of a vector $\bx$ or matrix $\bA$ is denoted by the corresponding subscripted lower-case letter $x_i$ or $a_{ij}$. All entries of the vector $\bzero$ equal $0$; $\bI$ indicates an identity matrix. The $^\top$ superscript indicates a vector transpose or a matrix transpose. The symbols $\bx_m$ and $\bA_m$ denote a sequence of vectors and a sequence of matrices with entries $x_{mi}$ and $a_{mij}$. The Euclidean norm and $\ell_1$ norm of a vector $\bx$ are denoted by $\|\bx\|$ and $\|\bx\|_1$. The spectral, Frobenius, and nuclear norms of a matrix $\bA$ are denoted by $\|\bA\|$, $\|\bA\|_F$, and $\|\bA\|_*$. For a smooth real-valued function $f(\bbeta)$, we write its first differential (row vector of partial derivatives) as $df(\bbeta)$. The gradient $\nabla f(\bbeta)$ is the transpose of $df(\bbeta)$. The directional derivative of $f(\bbeta)$ in the direction $\bv$ is written $d_{\bv}f(\bbeta)$. When the function $f(\bbeta)$ is differentiable, $d_{\bv}f(\bbeta)=df(\bbeta)\bv$. The second differential (Hessian matrix) of $f(\bbeta)$ is $d^2f(\bbeta)$.

\subsection{M Estimation \label{Sectionl2e}}

In M estimation \cite{de2021review,huber1974numerical,huber2011robust} based on residuals, one minimizes a function 
\begin{eqnarray*}
f(\bbeta) & = & \sum_{i=1}^n \rho(r_i) 
\end{eqnarray*}
of the residuals $r_i= y_i-\bx_i^\top\bbeta$ with $\rho(r)$ bounded below. We have already dealt with the choice $\rho(r)=|r|$ in LAD regression. In general, one can replace $\rho(r)$ by its Moreau envelope, leverage the Moreau majorization, and immediately invoke ordinary least squares. Many authors advocate estimating $\bbeta$ by minimizing a function
\begin{eqnarray*}
f(\bbeta) & = & \sum_{i=1}^n \rho(r_i^2) 
\end{eqnarray*}
of the squared residuals. Conveniently, this functional form is consistent with log-likelihoods under an elliptically symmetric distribution \cite{lange1993normal}. If $\rho(s)$ is increasing, differentiable, and concave,  then 
\begin{eqnarray*}
\rho(s) & \le &  \rho(s_m)+\rho'(s_m)(s-s_m)
\end{eqnarray*}
for all $s$. This plays out as the majorization
\begin{eqnarray*}
f(\bbeta) & \le &  \sum_{i=1}^n w_{mi}(y_i-\bx_i^\top\bbeta)^2
+c_m,
\end{eqnarray*}
with weights $w_{mi} = \rho'[(y_i-\bx_i^\top\bbeta_m)^2]$ and an irrelevant constant $c_m$ that depends on $\bbeta_m$ but not on $\bbeta$. This surrogate can be minimized
by deweighting and invoking ordinary least squares.  The Geman-McClure choice $\rho(s)= \frac{s}{1+s}$ and the Cauchy choice $\rho(s)=\log(1+s)$ for $s \ge 0$ show that the intersection of the imposed conditions is non-empty. 

A prime example \cite{chi2022user,heng2023robust} of robust M estimation revolves around the $\text{L}_2\text{E}$ objective \cite{scott2001parametric,liu2023sharper}
\begin{eqnarray*}
\label{L2E-loss}
f(\bbeta, \tau) & = & \frac{\tau}{2 \sqrt{\pi}}-\frac{\tau}{n} \sqrt{\frac{2}{\pi}} \sum_{i=1}^{n} e^{-\frac{\tau^{2}r_{i}^{2}}{2}}.
\end{eqnarray*}
This loss is particularly attractive because it incorporates a precision parameter $\tau$ as well as the regression coefficient vector $\bbeta$. Because $-e^{-s}$ is concave, one can construct the MM surrogate 
\begin{eqnarray*}
\label{majorization-r}
-e^{-\frac{\tau^{2}r_{i}^{2}}{2}} &\leq& -e^{-\frac{\tau^2r_{mi}^2}{2}} +\frac{\tau^2}{2} e^{-\frac{\tau^2 r_{mi}^2}{2}}\left(r_i^2 - r_{mi}^2\right)
\end{eqnarray*}
at iteration $m$ with $\tau$ fixed. The resulting weighted least squares surrogate can be further deweighted as explained in Section \ref{conversion_section}. Under deweighting it morphs into the surrogate
\begin{equation*}
g(\bbeta \mid \bbeta_m) = \frac{\tau^3}{2n} \sqrt{\frac{2}{\pi}} 
  \sum_{i=1}^n \left(\tilde{y}_{mi} - \bx_i^\top \bbeta\right)^2+c_m,
\end{equation*}
where $\tilde{y}_{mi} = w_{mi}y_i + (1-w_{mi})\bx_i^\top \bbeta_m$ is the shifted response and $w_{mi} = e^{-\tau^2 r_{mi}^2/2}$ is the case weight. This ordinary least squares surrogate allows one to recycle a single Cholesky decomposition across all iterations.  For $\bbeta$ fixed, one can update the precision parameter $\tau$ by gradient descent \cite{heng2023robust} or an approximate Newton's method \cite{liu2023sharper}. Both methods rely on backtracking.

\subsection{Quantile Regression \label{quantile_regression_section1}}

Quantile regression \cite{koenker1978regression} is a special case of M estimation with the objective $f(\bbeta) = \frac{1}{n}\sum_{i=1}^n \rho_q(y_i-\bx_i^\top \bbeta)$, where
\begin{eqnarray}
\rho_q(r) = \Big(q-\frac{1}{2}\Big)r+\frac{1}{2}|r| \nonumber \label{check_function}
\end{eqnarray}
is called the check function. The argument $r$ here suggests a residual. When $q=\frac{1}{2}$, quantile regression reduces to LAD regression. The check function is not differentiable at $r=0$. To smooth the objective (\ref{check_function}), one can take either its Moreau envelope or the Moreau envelope of just the nonsmooth part $\frac{1}{2}|r|$. In the former case, elementary calculus shows that the minimum is attained at
\begin{eqnarray*}
\prox_{\mu \rho_q} (r) & = & 
\begin{cases} r - q\mu, & r \ge q\mu \\ 0, & -(1-q)\mu < r < q \mu \\
r+(1-q) \mu, & r \le -(1-q)\mu \, . \end{cases}
\end{eqnarray*}
It follows that
\begin{equation*}
M_{\mu \rho_q}(r) = 
\begin{cases} 
  qr - \tfrac{\mu}{2}q^2,          & r \ge q\mu, \\
  -(1-q)r - \tfrac{\mu}{2}(1-q)^2, & r \le -(1-q)\mu, \\
  \tfrac{1}{2\mu}r^2,               & \text{otherwise.}
\end{cases}
\end{equation*}
A benefit of replacing $\rho_q(r)$ by $M_{\mu \rho_q}(r)$ is that, by definition, $M_{\mu \rho_q}(r)$ satisfies the Moreau majorization (\ref{Moreau_majorization}) at the anchor point $r_m$. This majorization can be improved, but it is convenient because it induces the uniform majorization
\begin{equation}\label{wholemoreau}
\sum_{i=1}^n M_{\mu \rho_q}(r_i)
 \le \sum_{i=1}^n \left[\rho_q(z_{mi})+\frac{1}{2\mu}(r_i-z_{mi})^2\right],
\end{equation}
where $z_{mi} = \prox_{\mu \rho_q}(r_{mi})$, which is an unweighted sum of squares in the residuals $r_i=y_i-\bx_i^\top\bbeta$.

Convolution approximation offers yet another avenue for majorization
\cite{fernandes2021smoothing,he2023smoothed,kaplan2017smoothed,tan2022high,whang2006smoothed}. If we select a well-behaved kernel $k(x) \ge 0$ with total mass $\int k(x) dx =1$, then the convolution 
\begin{eqnarray*}
C_{\mu g}(y) & = & \frac{1}{\mu} \int g(x)k\Big(\frac{y-x}{\mu}\Big)dx \\
& = &
\frac{1}{\mu} \int g(y-x)k\Big(\frac{x}{\mu}\Big)dx
\end{eqnarray*}
approximates $g(y)$ as the bandwidth $\mu \downarrow 0$. Given $g(x)=|x|$ and the symmetric uniform kernel $k(x)=\frac{1}{2}1_{\{|x| \le 1\}}$, the necessary integral 
\begin{eqnarray*}
C_{\mu |\cdot |}(y) & = & 
\frac{1}{2\mu}\int_{-\mu}^\mu |y-x| dx \\ & = &
\mu \begin{cases} \frac{1}{2}\Big[1+\Big(\frac{y}{\mu}\Big)^2\Big] & |y| \le \mu \\
\frac{|y|}{\mu} & |y| > \mu
\end{cases} 
\end{eqnarray*}
is straightforward to calculate. This is just the Moreau envelope $M_{\mu |\cdot|}(y)$ shifted upward by $\frac{\mu}{2}$. This approximate loss is now optimized via the Moreau majorization
\begin{align}\label{smoothedquantileloss}
&\Big(q-\tfrac{1}{2}\Big)r+\tfrac{1}{2}C_{\mu|\cdot|}(r) \notag\\
\leq\; &\frac{1}{4\mu}\big[r - z_{m} + (2q-1)\mu\big]^2+c_m,
\end{align}
where $c_m$ is an irrelevant constant and $z_{m} =  \prox_{\mu | \cdot |}(r_{m})$. Although our numerical experiments in quantile estimation feature the convolution-smoothed quantile loss \eqref{smoothedquantileloss}, the choice \eqref{wholemoreau} works equally well. What is more important is the transformation of smoothed quantile regression into a sequence of ordinary least squares problems with shifted responses. 

\subsection{Sparsity and Low Matrix Rank}
The $\ell_0$-norm, the rank function, and their corresponding constraint sets are helpful in inducing sparsity and low matrix rank. Directly imposing these nonconvex penalties or constraint sets can lead to intractable optimization problems. A common middle ground is to use convex surrogates, such as the $\ell_1$-norm and the nuclear norm, but these induce shrinkage bias and a surplus of false positives in support recovery. The Moreau envelopes of these nonconvex penalties offer an attractive alternative that leverages differentiability and majorization. 

The $\ell_0$ norm has Moreau envelope \cite{beck2017first}
\begin{eqnarray*}
M_{\mu \|\cdot \|_0}(\bbeta) & = & \underset{\bnu}\min \Big[\|\bnu\|_0+\frac{1}{2\mu}\|\bbeta-\bnu\|_2^2\Big] \\
& = & \sum_{j=1}^p \min_{\nu_j} \Big[1_{\{\nu_j \ne 0\}}+\frac{1}{2\mu}(\beta_j-\nu_j)^2\Big] \\
& = & \sum_{j=1}^p \begin{cases}
1, & \frac{1}{2}\beta_j^2 \ge \mu \\ \frac{1}{2\mu}\beta_j^2, 
& \frac{1}{2}\beta_j^2 < \mu \,
\end{cases}.
\end{eqnarray*}
The corresponding proximal map sends $\beta_j$ to itself when $\frac{1}{2}\beta_j^2 \ge \mu$ and to $0$ when $\frac{1}{2}\beta_j^2 < \mu$. The Moreau envelope of the $0/\infty$ indicator of the sparsity set $S_k=\{\bbeta: \|\bbeta\|_0 \le k\}$ is determined by projection onto $S_k$. Elementary arguments show that 
\begin{eqnarray*}
M_{\mu \delta_{S_k}}(\bbeta) & = & \min_{\bnu} \Big[\delta_{S_k}(\bnu)+\frac{1}{2\mu}\|\bbeta-\bnu\|_2^2\Big] \\
& = & \frac{1}{2\mu}\sum_{j=1}^{p-k} \beta_{(j)}^2,
\end{eqnarray*}
where $\beta_{(1)},\ldots,\beta_{(p-k)}$ denote the $p-k$ smallest $\beta_j$ in magnitude \cite{beck2017first}. The Moreau majorization defined by equation (\ref{Moreau_majorization}) is available in both examples. In general, the Moreau envelope of the $0/\infty$ indicator of a set $S$ is $\frac{1}{2\mu}\dist(\bbeta, S)^2$, where $\dist(\bbeta,S)$ is the Euclidean distance  from $\bbeta$ to $S$. Other nonconvex sparsity inducing penalties $p(\beta)$, such as the SCAD penalty \cite{SCAD} and the MCP penalty \cite{zhang2010nearly}, also have straightforward Moreau envelopes \cite{proximal}.

For matrices, low-rank is often preferred to sparsity in estimation and inference. Let $R_k$ denote the set of $p \times c$ matrices of rank $k$ or less. The Moreau envelope of the $0/\infty$ indicator of $R_k$ turns out to be  $\frac{1}{2\mu}\dist(\bB,R_k)^2 =  \frac{1}{2\mu}\sum_{j>k}\sigma_j^2$, where the $\sigma_j$  are the singular values of $\bB$ in descending order. If $\bB$ has singular value decomposition $\bU\bSigma\bV^\top$, then the corresponding proximal map takes $\bB$ to $\bU\bSigma_k\bV^\top$, where $\bSigma_k$ retains the 
$k$ largest singular values and sets the rest to zero.  This result is just the content of the Eckart-Young theorem.  The  function $\rank(\bB)$ has Moreau envelope \cite{hiriart2013eckart}
 \begin{eqnarray*}
M_{\mu \, \rank}(\bB) & = & M_{\mu \|\cdot \|_0}(\bsigma)  =  \sum_i \begin{cases}
1, & \frac{1}{2}\sigma_i^2 \ge \mu \\ \frac{1}{2\mu}\sigma_i^2, & \frac{1}{2}\sigma_i^2 < \mu \, 
\end{cases}
\end{eqnarray*}
The map $\prox_{\mu \,\rank}(\bB) = \bU\bSigma_\mu\bV^\top$ retains $\sigma_j$ when $\frac{1}{2}\sigma_j^2 \ge \mu$ and sets it to $0$ otherwise. In other words, the proximal operator hard thresholds the singular values.  In both cases the Moreau majorization \eqref{Moreau_majorization} is available.

The nuclear norm $\| \bB \|_*=\sum_i \sigma_i$ \cite{recht2010guaranteed} is a widely adopted convex penalty in low-rank estimation. Its proximal map soft thresholds all $\sigma_i$ by subtracting $\mu$ from each, then replacing negative values by $0$, and finally setting $\prox_{\mu \|\cdot\|_*}(\bB)= \bU\tilde{\bSigma}\bV^\top$, where $\tilde{\bSigma}$ contains the thresholded values. The Moreau envelope $M_{\mu\| \cdot \|_*}(\bB)$ serves as a smooth-approximation to $\|\bB\|_*$ and benefits from the Moreau majorization \eqref{Moreau_majorization}.

\subsection{Deweighting Weighted Least Squares}
\label{conversion_section}

To remove the weights and shift the responses in least squares, \cite{heiser1987correspondence} and \cite{kiers1997weighted} suggest a crucial deweighting majorization. At iteration $m$, assume that the weights satisfy $0 \le w_{mi} \le 1$ and that $\mu_i$ abbreviates the mean $\mu_i(\bbeta)$ of case $i$. Then
\begin{eqnarray*}
&&w_{mi}(y_i - \mu_i)^2\\
&=& w_{mi}\bigl[(y_i-\mu_{mi}) + (\mu_{mi}-\mu_i)\bigr]^2 \\
&=& w_{mi}(y_i-\mu_{mi})^2
    + w_{mi}(\mu_{mi}-\mu_i)^2 \\
&&+\; 2w_{mi}(y_i-\mu_{mi})(\mu_{mi}-\mu_i) \\
&=& w_{mi}^2(y_i-\mu_{mi})^2
    + w_{mi}(\mu_{mi}-\mu_i)^2 \\
&&+\; 2w_{mi}(y_i-\mu_{mi})(\mu_{mi}-\mu_i)
    + d_{mi} \\
&\le& w_{mi}^2(y_i-\mu_{mi})^2
     + (\mu_{mi}-\mu_i)^2 \\
&&+\; 2w_{mi}(y_i-\mu_{mi})(\mu_{mi}-\mu_i)
     + d_{mi} \\
&=& \bigl[w_{mi}(y_i-\mu_{mi}) + (\mu_{mi}-\mu_i)\bigr]^2
    + d_{mi} \\
&=& \bigl[w_{mi}y_i + (1-w_{mi})\mu_{mi} - \mu_i\bigr]^2
    + d_{mi},
\end{eqnarray*}
where $d_{mi} = w_{mi}(1-w_{mi})(y_i-\mu_{mi})^2 \ge 0$
and $d_m = \sum_{i=1}^n d_{mi}$ does not depend on $\bmu$. Equality holds, as it should, when $\mu_i=\mu_{mi}$. Derivation of the surrogate requires the assumption $0 \le w_{mi} \le 1$. Because rescaling does not alter the minimum point, one can divide $w_{mi}$ by $w_{\max} = \max_j w_{mj}$ at iteration $m$ to ensure $0 \le w_{mi} \le 1$.

\subsection{Quadratic Upper and Lower Bound Principles \label{quad_bound_section}}
The quadratic upper bound principle applies to functions $f(\bbeta)$ with bounded curvature \cite{bohning1988monotonicity}. Given that $f(\bbeta)$ is twice differentiable, we seek a matrix $\bH$ satisfying $\bH \succeq d^2f(\bbeta)$ and $\bH \succ {\bf 0}$ in the sense that $\bH-d^2f(\bbeta)$ is positive semidefinite for all $\bbeta$ and $\bH$ is positive definite.  The quadratic bound principle then amounts to the majorization
\begin{eqnarray*}
f(\bbeta) & = & f(\bbeta_m) + df(\bbeta_m)(\bbeta-\bbeta_m) \nonumber 
 +\,(\bbeta-\bbeta_m)^\top \\ &&\int_0^1 d^2f[\bbeta_m+t(\bbeta-\bbeta_m)](1-t)\,dt\,(\bbeta-\bbeta_m) \nonumber \\
& \le & f(\bbeta_m) + df(\bbeta_m)(\bbeta-\bbeta_m) \\&&+\frac{1}{2}(\bbeta-\bbeta_m)^\top\bH(\bbeta-\bbeta_m) \label{MMconvex_4} \\
& = & g(\bbeta \mid \bbeta_m). \nonumber
\end{eqnarray*}
The quadratic surrogate is exactly minimized by the update (\ref{MMquadratic_update}) with $d^2g(\bbeta_m \mid \bbeta_m)$ replaced by $\bH$. The quadratic lower bound principle involves minorization and subsequent maximization. 

The quadratic upper bound principle applies to logistic regression and its generalization to multinomial regression \cite{bohning1992multinomial,krishnapuram2005sparse}. In multinomial regression, items are drawn from $c$ categories numbered $1$ through $c$. Logistic regression is the special case $c=2$. If $y_i$ denotes the response of sample $i$, and $\bx_i$ denotes a corresponding predictor vector, then the probability $w_{ij}$ that $y_i=j$ is
\begin{eqnarray*}
w_{ij}(\bB) & = & \begin{cases} \frac{e^{r_{ij}}}{1+\sum_{k=1}^{c-1} e^{r_{ik}}}, & 1 \le j < c \\
\frac{1}{1+\sum_{k=1}^{c-1} e^{r_{ik}}}, & j=c 
\end{cases}
\end{eqnarray*}
where $r_{ij} =\bx_i^\top \bbeta_j$. To estimate the regression coefficient vectors $\bbeta_j$ over $n$ independent trials, we must find a quadratic lower bound for the log-likelihood
\begin{eqnarray}\label{multinomiallikelihood}
\quad \quad\mathcal{L}(\bB) & = & 
\sum_{i=1}^n \begin{cases} r_{iy_i} - \ln (1+\sum_{j=1}^{c-1} e^{r_{ij}}), & 1 \le y_i < c \\
- \ln (1+\sum_{j=1}^{c-1} e^{r_{ij}}), &y_i=c 
\end{cases}
\end{eqnarray}
Here $\bB$ is the matrix with $j$th column $\bbeta_j$. It suffices to find a quadratic lower bound for each sample, so we temporarily drop the subscript $i$ to simplify notation. 

The log-likelihood has directional derivative 
\begin{eqnarray*}
d_{\bV}\mathcal{L}(\bB)  & = & 1_{\{y < c\}} \bx^\top \bv_{y}  -\frac{\sum_{j=1}^{c-1} e^{r_j}\bx^\top\bv_j}{1+\sum_{j=1}^{c-1} e^{r_j}} \\
& = & [(\by-\bw) \otimes \bx]^\top \vec(\bV)
\end{eqnarray*}
for the perturbation $\bV$ of $\bB$ with $j$th column $\bv_j$. In the Kronecker product representation of $d_{\bV}\mathcal{L}(\bB)$, $\by$ is now an indicator vector of length $c-1$ with a $1$ in entry $j$ if the sample belongs to category $j$, and $\bw$ is the weight vector with value $w_{j} = \frac{e^{r_{j}}}{1+\sum_{k=1}^{c-1} e^{r_{k}}}$ in entry $j<c$. The quadratic form associated with the second differential $d^2\mathcal{L}(\bB)$ is
\begin{eqnarray*}
d^2_{\bV} \mathcal{L}(\bB) & = & -\frac{\sum_{j=1}^{c-1} e^{r_j}(\bx^\top\bv_j)^2}{1+\sum_{j=1}^{c-1} e^{r_j}}
+ \frac{(\sum_{j=1}^{c-1} e^{r_j}\bx^\top\bv_j)^2}{(1+\sum_{j=1}^{c-1} e^{r_j})^2} .
\end{eqnarray*} 
Across all cases \cite{bohning1992multinomial} derives the lower bound
\begin{eqnarray*}
d^2_{\bV} \mathcal{L}(\bB)
& \ge & -\frac{1}{2} \vec(\bV)^\top \left[ \bE\otimes \left(\sum_{i=1}^n \bx_i\bx_i^\top\right) \right] \vec(\bV),
\end{eqnarray*}
where the $(c-1) \times (c-1)$ matrix $\bE=\frac{1}{2}(\bI - \frac{1}{c}\bone \bone^\top)$
has explicit inverse $\bE^{-1} = 2(\bI+\bone \bone^\top)$.
The full-sample multinomial log-likelihood is then minorized by the quadratic
\begin{eqnarray}\label{quadboundmultinomial}
\mathcal{L}(\bB) & \ge & \mathcal{L}(\bB_m) 
+ \sum_{i=1}^n [(\by_i-\bw_{mi}) \otimes \bx_i]^\top \vec(\bDelta) \nonumber \\
& & -\frac{1}{2} \vec(\bDelta)^\top 
\left[ \bE\otimes \bX^\top\bX \right] \vec(\bDelta),\nonumber
\end{eqnarray}
where $\bDelta = \bB - \bB_m$. Maximizing the surrogate yields the update
\begin{eqnarray*}
\vec(\bDelta_{m+1}) & = & 
\left[ \bE \otimes \bX^\top\bX \right]^{-1} 
\sum_{i=1}^n (\by_i-\bw_{mi}) \otimes \bx_i \\
& = & 
\left[ \bE^{-1}\otimes (\bX^\top \bX)^{-1} \right] 
\sum_{i=1}^n (\by_i-\bw_{mi}) \otimes \bx_i,
\end{eqnarray*}
or in matrix form
\begin{eqnarray}\label{multinomialupdate}
\bDelta_{m+1} & = & (\bX^\top \bX)^{-1} \bX^\top (\bY - \bW_m)\bE^{-1},
\end{eqnarray}
so that $\bB_{m+1} = \bB_m + \bDelta_{m+1}$. Here $\bW_m\in \Real^{n\times(c-1)}$
conveys the weights for the first $c-1$ categories and the rows of
$\bY\in \Real^{n\times(c-1)}$ are indicator vectors conveying category membership.
In the special case of logistic regression, the solution reduces to
\begin{eqnarray*}
\bbeta_{m+1} & = & \bbeta_{m} + 4(\bX^\top \bX)^{-1} \bX^\top(\by -\bw_m).
\end{eqnarray*}
We also consider low-rank multinomial regression driven by the penalized loss $-\frac{1}{n}\mathcal{L}(\bB) + \lambda M_{\mu\|\cdot\|_*}(\bB)$. Applying the minorization \eqref{quadboundmultinomial} and the Moreau majorization yields the MM surrogate
\begin{align*}
    g(\bB \mid \bB_m) 
    &= \frac{1}{n}\vec[\bX^\top(\bW_m - \bY)]^\top \vec(\bDelta) \\
    &\quad + \frac{1}{2n}\vec(\bDelta)^\top 
    \bigl[\bE \otimes \bX^\top\bX\bigr] \vec(\bDelta) \\
    &\quad + \frac{\lambda}{2\mu}\bigl\|\bB - 
    \prox_{\mu\|\cdot\|_*}(\bB_m)\bigr\|_F^2 + c_m,
\end{align*}
whose stationary condition is given by
\begin{eqnarray*}
\bzero & = & \vec(\bC_m)-\left(\frac{1}{n}\bE \otimes \bX^\top \bX + \frac{\lambda}{\mu} \bI\right) \vec(\bDelta),
\end{eqnarray*}
with $\bC_m = \frac{1}{n}\bX^\top (\bY - \bW_m) + \frac{\lambda}{\mu}(\prox_{\mu\|\cdot\|_*}(\bB_m)-\bB_m)$ and $\bDelta=\bB-\bB_m$. This condition is the same as the two equivalent matrix equations
\begin{eqnarray*}
\frac{1}{n}\bX^\top \bX \bDelta \bE + \frac{\lambda}{\mu} \bDelta & = & \bC_m \\
\frac{1}{n}\bX^\top \bX \bDelta + \frac{\lambda}{\mu} \bDelta \bE^{-1} & = & \bC_m \bE^{-1}.
\end{eqnarray*}
The second of these is a Sylvester equation \cite{bartels1972algorithm} for the matrix $\bDelta$. It can be solved by extracting the spectral decompositions $\bX^\top\bX = \bU\bSigma_1\bU^\top$ and $\bE^{-1} = \bV\bSigma_2\bV^\top$, left multiplying the equation by $\bU^\top$, and right multiplying the equation by $\bV$. If in the result
\begin{eqnarray*}
\frac{1}{n}\bSigma_1\bU^\top \bDelta \bV + \frac{\lambda}{\mu}\bU^\top \bDelta \bV \bSigma_2 & = & 
\bU^\top \bC_m \bV \bSigma_2
\end{eqnarray*}
we treat $\bU^\top \bDelta \bV$ as a new variable $\bZ=(z_{ij})$, then this simpler Sylvester equation has the entry-by-entry solution
\begin{eqnarray*}
z_{ij} & = & \frac{(\bU^\top \bC_m \bV)_{ij} \sigma_{2j}}{\frac{1}{n}\sigma_{1i} + \frac{\lambda}{\mu} \sigma_{2j}},
\end{eqnarray*}
where $\sigma_{1i}$ and $\sigma_{2j}$ are the diagonal entries of $\bSigma_1$ and $\bSigma_2$.
Overall, this highly efficient way of updating $\bB$ requires just two spectral decompositions, which can be recycled across all iterations and all values of $\lambda$ and $\mu$. Computation of the proximal map $\prox_{\mu\|\cdot\|_*}(\bB_m)$ must still be performed at each iteration.

\subsection{Integration with the Proximal Distance Principle}

The proximal distance principle is designed to minimize a loss function $f(\bbeta)$ subject to $\bbeta \in C$, where $C$ is a nonempty closed set \cite{keys2019proximal,landeros2022extensions,landeros2023inaugural}. For instance, $C$ could be the sparsity set $\{ \bbeta\in \Real^p: \| \bbeta\|_0 \le k\}$. The general idea of the proximal distance method is to approximate the solution by minimizing the penalized loss $f(\bbeta)+\frac{\lambda}{2}\dist(\bbeta,C)^2$ for $\lambda>0$ large. The squared Euclidean distance is majorized by the spherical quadratic $\frac{\lambda}{2}\|\bbeta-P_C(\bbeta_m)\|_2^2$, where $P_C(\bbeta)$ denotes the Euclidean projection of $\bbeta$ onto $C$. If $f(\bbeta)$ is an ordinary sum of squares, then the surrogate $f(\bbeta)+\frac{\lambda}{2}\|\bbeta-P_C(\bbeta_m)\|^2$ remains within this realm. Otherwise, it may be that $f(\bbeta)$ is majorized by an ordinary sum of squares $g(\bbeta \mid \bbeta_m)$. The overall surrogate $g(\bbeta \mid \bbeta_m)+\frac{\lambda}{2}\|\bbeta-P_C(\bbeta_m)\|^2$ is then also an ordinary sum of squares.

\begin{algorithm}[h]
\caption{Proximal Distance Algorithm}
\label{alg:proximal_distance}
\begin{algorithmic}[1]
\REQUIRE loss function $f$, constraint set $C$, initial point $\bbeta_0$,
         annealing rate $\rho > 1$, tolerances $\epsilon_{\mathrm{in}},
         \epsilon_{\mathrm{out}} > 0$
\ENSURE approximate constrained minimizer $\bbeta$
\STATE $\bbeta \leftarrow \bbeta_0$, $\lambda \leftarrow \lambda_0$
\WHILE{$\|\bbeta - P_C(\bbeta)\|_2 \geq \epsilon_{\mathrm{out}}$}
    \REPEAT
        \STATE $\bbeta \leftarrow \argmin_{\bbeta'}\; g(\bbeta' \mid \bbeta) 
               + \tfrac{\lambda}{2}\|\bbeta' - P_C(\bbeta)\|_2^2$
    \UNTIL{$\text{convergence criterion} < \epsilon_{\mathrm{in}}$}
    \STATE $\lambda \leftarrow \rho \cdot \lambda$
\ENDWHILE
\RETURN $\bbeta$
\end{algorithmic}
\end{algorithm}

In the proximal distance algorithm, simply setting $\lambda$ equal to a large 
constant can cause the procedure to converge prematurely and prevent adequate 
exploration of parameter space. Although the converged value of $\bbeta$ will 
be close to the constraint set $C$, the influence of the loss will be slight. 
The remedy is to start with a small value of $\lambda$ and gradually increase 
it at a slow geometric rate to a final desired level.  For a given 
$\lambda$, we leverage distance majorization to optimize the penalized loss.
Once the convergence criterion is met, we increment $\lambda$ and start anew. Algorithm~\ref{alg:proximal_distance} summarizes the process. Our 
previous papers \cite{keys2019proximal,landeros2022extensions} provide 
essential algorithmic details such as annealing schedules and stopping 
criteria. In this paper, we take a slightly different approach in dealing 
with the Moreau envelope parameter $\mu$ governing the smoothed 
$\ell_0$ and nuclear norm penalties. Rather than adopting the inner/outer 
iteration structure of Algorithm~\ref{alg:proximal_distance}, we either 
hold $\mu$ fixed or anneal it simultaneously with the iterative 
updates.

Sparse quantile regression provides a concrete example of the proximal distance 
principle in action. After distance penalization, the smoothed quantile loss 
(\ref{smoothedquantileloss}) is majorized by the surrogate function
\begin{eqnarray}\label{eq:surrogatesq2}
g(\bbeta \mid\bbeta_m) & = & \frac{1}{4n\mu} \sum_{i=1}^n 
[y_i - z_{mi} + (2q-1)\mu - \bx_i^\top\bbeta ]^2 \nonumber \\
&&+ \frac{\lambda}{2} \| \bbeta-P_{S_k}(\bbeta_m)\|_2^2+ c_m.
\end{eqnarray}
The stationary condition is then
\begin{eqnarray*}
\Big(\frac{1}{2n\mu}\bX^\top\bX +\lambda\bI_p\Big) \bbeta 
& = & \frac{1}{2n\mu}\bX^\top\tilde{\by} + \lambda P_{S_k}(\bbeta_m),
\end{eqnarray*}
where $\tilde{y}_i = y_i - z_{mi} + (2q-1)\mu$ are the shifted responses.
Let $M = \frac{1}{2n\mu}$ and $\bb_m = M\bX^\top\tilde{\by} 
+ \lambda P_{S_k}(\bbeta_m)$ for brevity. Given the pre-computed spectral 
decomposition $\bX^\top\bX = \bU\bLambda\bU^\top$ of the $p \times p$ Gram 
matrix, the normal equations have the explicit solution
\begin{eqnarray*}
\bbeta_{m+1} & = & \bU \big( M\bLambda + \lambda \bI_p \big)^{-1}
\bU^\top \bb_m.
\end{eqnarray*}
Although a spectral decomposition is more expensive to compute than a Cholesky 
decomposition, a single decomposition suffices across all values of $\lambda$. 
In contrast, the Cholesky decomposition must be recomputed each time $\lambda$ 
changes.  When $p \gg n$, we instead work with the spectral decomposition 
$\bX\bX^\top = \bQ\bPhi\bQ^\top$ of the $n\times n$ matrix $\bX\bX^\top$. 
The Woodbury identity gives
\begin{eqnarray*}
&&\big(M\bX^\top\bX + \lambda\bI_p\big)^{-1}\\
& = & \frac{1}{\lambda}\bI_p 
- \frac{M}{\lambda^2}\bX^\top\bQ
\Big(\bI_n+\frac{M}{\lambda}\bPhi\Big)^{-1}\bQ^\top\bX
\end{eqnarray*}
and the update becomes
\begin{eqnarray*}
\bbeta_{m+1} & = & \frac{1}{\lambda}\bb_m
- \frac{M}{\lambda^2}\bX^\top\bQ
\Big(\bI_n+\frac{M}{\lambda}\bPhi\Big)^{-1}\bQ^\top\bX\bb_m.
\end{eqnarray*}
Since $\bX\bX^\top$ is $n \times n$, the spectral decomposition costs $O(n^3)$ 
rather than $O(p^3)$, and the matrix--vector products drop from $O(p^2)$ to $O(np)$, yielding substantial savings 
in high-dimensional problems.

Sparse quantile regression can also be achieved by adding the Moreau envelope 
penalty $\lambda M_{\alpha \|\cdot \|_0}(\bbeta)$ to the loss. (Observe that 
$\mu$ and $\alpha$ are distinct smoothing parameters; nothing prevents us from 
equating them.) The surrogate function (\ref{eq:surrogatesq2}) remains the same, 
provided we substitute the proximal map $\prox_{\alpha\|\cdot\|_0}(\bbeta_m)$ 
for the projection $P_{S_k}(\bbeta_m)$ and the constant $\frac{\lambda}{2\alpha}$ 
for the constant $\frac{\lambda}{2}$. If $\bbeta_\lambda$ denotes the minimum 
point of the penalized loss, then the curve $\lambda \mapsto 
M_{\alpha \|\cdot \|_0}(\bbeta_\lambda)$ now becomes an object of interest. 
This curve should smoothly approximate a jump function with downward jumps 
of $1$ as $\lambda$ increases.

\subsection{Acceleration Strategies}\label{acceleration}

Conversion of weighted to ordinary least squares comes at the expense of an increased number of iterations until convergence. This adverse consequence can be mitigated by various acceleration strategies. 

\paragraph*{Nesterov acceleration} Nesterov acceleration \cite{nesterov1983method} is an important acceleration technique that is easier to implement than to understand. In practice, one maintains both the current iterate $\bbeta_m$ and the previous iterate $\bbeta_{m-1}$. The shifted point $\balpha_m=\bbeta_m+\frac{m-1}{m+2}(\bbeta_m-\bbeta_{m-1})$ is then taken as the base point in majorization-minimization. The next point $\bbeta_{m+1}$ is no longer guaranteed to reduce the objective. If the descent property fails, we take the standard precaution of restarting the Nesterov acceleration at $m=1$.

\paragraph*{Step-doubling} 
Step-doubling \cite{de1980multidimensional,van2016gensvm} is another 
simple yet effective acceleration device. 
Let $\bbeta^*$ denote the minimizer of the surrogate 
$g(\bbeta \mid \bbeta_m)$. Rather than taking $\bbeta^*$ as the next 
iterate, step doubling proposes
\begin{eqnarray*}
\bbeta_{m+1} & = & \bbeta_m + 2(\bbeta^* - \bbeta_m) 
\; = \; 2\bbeta^* - \bbeta_m.
\end{eqnarray*}
It is worth noting that, if the surrogate function is quadratic, then step-doubling never increases the objective. The key distinction between step-doubling and Nesterov is that Nesterov extrapolates before the MM step ($\bbeta_{m+1}=T( \bbeta_m + \tfrac{m-1}{m+2}(\bbeta_m - \bbeta_{m-1}))$), whereas step-doubling extrapolates after the MM step ($\bbeta_{m+1} = 2T(\bbeta_m) - \bbeta_m$). Here $T(\bbeta)
    = \bbeta - \bH^{-1}\nabla f(\bbeta)$ is the quadratic MM mapping.

\paragraph*{Smoothing Constant Annealing}
Replacing nonsmooth loss functions and penalties with their Moreau envelopes produces a smoother optimization problem. Larger values of $\mu$ yield smoother approximations at the cost of greater deviation from the original objective.  Sending $\mu$ to $0$ recovers the original function at the cost of increasing computational difficulty. To balance statistical accuracy with computational efficiency, we initialize $\mu$ at a large value $\mu_{\max}$ to capitalize on rapid early progress and then gradually decrease $\mu$ toward zero. When $\mu$ is exceedingly small, the curvature of the Moreau envelope near the origin becomes extremely steep, causing the algorithm to make little or no progress. Consequently, we impose a lower bound $\mu_{\min} > 0$, below which $\mu$ is not allowed to go. The value $\mu_{\min}$ governs the statistical properties of our smoothed estimators.

Other popular acceleration strategies include Anderson acceleration \cite{anderson1965iterative} and SQUAREM \cite{varadhan2008simple}. In our experience, both are more complex to implement and less effective on the problems considered here. 

\subsection{Convergence Theory \label{convergence_theory_section}}

As a prelude to establishing the convergence of our MM algorithms, one needs to prove that an optimal point actually exists. This issue is usually tackled in minimization by showing that the objective $f(\bbeta)$ has compact sublevel sets $\{\bbeta: f(\bbeta) \le c\}$. This notion is implied by the coerciveness condition $\lim_{\|\bbeta\|\to \infty}f(\bbeta)=\infty$, which in turn is implied by strong convexity. In the absence of strong convexity, one can check that a finite convex function $f(\bbeta)$ is coercive by checking that it is coercive along all nontrivial rays $\{\bbeta \in \Real^p: \bbeta=\balpha+t\bv,\, t \ge 0\}$ emanating from some arbitrarily chosen base point $\balpha$ \cite{lange2016mm}. For instance in quantile regression, the objective $f(\bbeta)$ is convex and 
\begin{eqnarray*}
f(\bbeta+t\bv) & = & \sum_{i=1}^n \rho_q(y_i-\bx_i^\top \bbeta - t \bx_i^\top \bv).
\end{eqnarray*}
If $\bX$ has full rank and $\bv \ne \bzero$, then at least one inner product $\bx_i^\top\bv$ does not vanish. The corresponding term $\rho_q(y_i-\bx_i^\top \bbeta - t \bx_i^\top \bv)$ then tends to $\infty$ as $t$ tends to $\infty$. Given that all other terms are nonnegative, $f(\bbeta)$ is coercive.

As pointed out earlier, our algorithm maps assume the form
\begin{eqnarray*}
\mathcal{M}(\bbeta) & = & \bbeta - t_{\bbeta}\bH(\bbeta)^{-1} \nabla f(\bbeta), \label{algorithm_map}
\end{eqnarray*}
where the matrix $\bH(\bbeta) = d^2 g(\bbeta \mid \bbeta)$ is positive definite and $t_{\bbeta} \in (0,1)$ is a step-size multiplier chosen by backtracking. For our MM algorithms, the choice $t_{\bbeta}=1$ always decreases the objective. \cite{lange2013optimization} and \cite{xu2017generalized} prove the following proposition pertinent to our problems.
\begin{proposition} \label{backtracking_convergence}
Suppose the objective $f(\bbeta)$ has compact sublevel sets and $\bH(\bbeta)$ is continuous and positive definite. If the step size $t_{\bbeta}$ is selected by Armijo backtracking, then the limit points of the sequence $\bbeta_{m+1} = \mathcal{M}(\bbeta_{m})$ are stationary points of $f(\bbeta)$. Moreover, the set of limit points is compact and connected. If all stationary points are isolated, then the sequence $\bbeta_m$ converges to one of them.
\end{proposition}

Our examples omit backtracking and involve surrogates satisfying the uniform Lipschitz gradient condition
\begin{eqnarray}
\|\nabla g(\bgamma \mid \bbeta_m)-\nabla g(\bdelta \mid \bbeta_m) \| \label{lipschitz_uniform}
& \le & L \|\bgamma-\bdelta\|
\end{eqnarray}
for all $\bgamma$, $\bdelta$, and $\bbeta_m$. Proposition 8 of \cite{MMconverge} then offers the following simpler but more precise result.
\begin{proposition}
Let $f(\bbeta)$ be a coercive differentiable function majorized by a surrogate satisfying condition (\ref{lipschitz_uniform}). If $\bbeta_\infty$ denotes a minimum point of $f(\bbeta)$, then the iterates $\bbeta_m$ delivered by the corresponding MM algorithm satisfy the bound
\begin{eqnarray*}
\sum_{k=0}^m \| \nabla f(\bbeta_k) \|^2 & \le & 2L [f(\bbeta_0) - f(\bbeta_\infty)] . \label{sublinear_bound}
\end{eqnarray*}
It follows that $\lim_{m \to \infty} \| \nabla f(\bbeta_m) \|=0$. Furthermore, when $f(\bbeta)$ is continuously differentiable, any limit point of the sequence $\bbeta_m$ is a stationary point of $f(\bbeta)$.
\end{proposition}

Proposition 12 of \cite{MMconverge} dispenses with the isolated stationary point assumption of Proposition \ref{backtracking_convergence} and proves  convergence when $f(\bbeta)$ is coercive, continuous, and subanalytic and all $g(\bbeta \mid \bbeta_m)$ are continuous, $\mu$-strongly convex, and satisfy condition (\ref{lipschitz_uniform}) on the compact sublevel set $\{\bbeta: f(\bbeta) \le f(\bbeta_0)\}$. Virtually all well-behaved functions are subanalytic. For instance, the Moreau envelope of a convex lower-semicontinuous subanalytic function is subanalytic \cite{bolte2007lojasiewicz}; the function $\|\bbeta\|_0$ is semialgebraic and therefore subanalytic \cite{landeros2022extensions}.

For the convex problems, one can also attack global convergence by invoking monotone operator theory \cite{bauschke2011convex}. 
\begin{proposition}\label{thm:convexconverge}
For a convex problem with $\bH=d^2g(\bbeta \mid \bbeta)$ a fixed positive definite matrix, suppose that the set of stationary points is nonempty. Then the MM iterates 
\begin{eqnarray*}
\bbeta_{m+1} & = & \bbeta_m - \bH^{-1} \nabla f(\bbeta_m)
\end{eqnarray*}
converge to a minimum point. 
\end{proposition}
\begin{proof}
See Appendix \ref{proofs_appendix}.
\end{proof}

Proposition 10 of \cite{MMconverge} proves that MM  iterates converge at a linear rate in the best circumstances. 
\begin{proposition} \label{linear_conv_prop}
Let $f(\bbeta)$ be $\mu$-strongly convex and differentiable, and assume $g(\bbeta \mid \bbeta_m)$  satisfies the Lipschitz condition (\ref{lipschitz_uniform}). If the global minimum occurs at $\balpha$, then the MM iterates $\bbeta_m$ satisfy
\begin{eqnarray*}
f(\bbeta_{m}) - f(\balpha) & \le &
\left[1 - \left(\frac{\mu}{2L}\right)^{2}\right]^{m}
[f(\bbeta_{0}) - f(\balpha)],
\end{eqnarray*}
establishing linear convergence of $\bbeta_m$ to $\balpha$.
\end{proposition}
\begin{proof}
See Appendix \ref{proofs_appendix} for a slightly corrected proof.
\end{proof}

In many of our examples, the Hessian $d^2g(\bbeta \mid \bbeta)= c \bX^\top\bX$ for some $c>0$. The constants of Proposition \ref{linear_conv_prop} satisfy $L=c \sigma_1^2$ and $\mu = c \sigma_p^2$, where the $\sigma_i$ are the ordered singular values of $\bX$. The ratio $\frac{L}{\mu}$ is the condition number of $\bX$ relevant to the rate of convergence. Note that $c$ disappears in the condition number. In the proximal distance method, we encounter the Hessian $d^2g(\bbeta \mid \bbeta)= c \bX^\top\bX+\lambda \bI$, where the constant $c$ no longer cancels in the condition number. Although addition of $\lambda \bI$ improves the condition number pertinent to inner iterations, annealing of $\lambda$ towards $\infty$ in outer iterations slows overall convergence to the constrained minimum. Multinomial regression replaces the Gram matrix $\bX^\top\bX$ by the Kronecker product $ \left(\bI - \frac{1}{c}\bone \bone^\top\right)\otimes \left(\bX^\top\bX\right)$, whose singular values are products of the eigenvalues of the matrices $\bI - \frac{1}{c}\bone \bone^\top$  and $\bX^\top\bX$. Because these are $\frac{1}{c}$ and $1$ and the $\sigma_i^2$, respectively, the revised condition number turns out to equal $\frac{cL}{\mu}$.

\section{Numerical Experiments}\label{exp}

The numerical examples in this section demonstrate the versatility of the MM principle in deriving fast estimation algorithms. We consider quantile regression (ordinary and sparse), $\text{L}_2\text{E}$ regression (ordinary and isotonic), and multinomial regression (ordinary and reduced-rank). For sparse and low-rank regression, the penalties ignore intercept parameters. All computations were performed on a Mac with an M4 chip
and 32GB of RAM. We primarily employ restarted Nesterov acceleration to speed up convergence of our MM algorithms. For problems involving a smoothing parameter $\mu$, we also invoke the annealing strategy for $\mu$ described in Section \ref{ablation}. 
\subsection{Unpenalized Models}\label{sec:nonstructured}

We now compare the performance of our MM algorithms and popular competing algorithms for smoothed quantile regression, $\text{L}_2\text{E}$ regression, and multinomial regression. No constraints or penalties are imposed at this stage. We adopt a common protocol for generating the design matrix $\bX\in \Real^{n\times (p-1)}$. The aim here is to demonstrate the computational efficiency of the MM algorithms on large-scale data. Each row of $\bX$ is sampled from a multivariate normal distribution with mean $\bzero$ and covariance matrix $\bSigma$ with  entries $\sigma_{ij}=0.7^{|i-j|}$. An intercept column $\bone_n$ is then left-appended to $\bX$ to obtain a full design matrix $\tilde{\bX}\in \Real^{n\times p}$. Generally, we set the true coefficient vector to  $\bbeta^*=(1,0.1,0.1,\dots,0.1)^\top \in \Real^p$. The responses $\by$ require a different generation protocol for each problem class. Unless stated otherwise, the number of observations equals $n=100p$, where the number of predictors $p$ ranges over the set $\{100,200,\dots,1000\}$. We terminate our MM algorithms when the relative change in objective falls below $10^{-6}$. In the following paragraphs, we describe competing methods, their software implementations, and the further protocols for generating the response vector $\by$. 

\textbf{Quantile Regression} For convolution-smoothed quantile regression, we compare MM  to the \texttt{conquer} R package \cite{tan2022high,he2023smoothed}. The \texttt{conquer} package implements gradient descent with a Barzilai-Borwein step size. In calling \texttt{conquer}, we invoke the uniform kernel with the default bandwidth $\mu=\max\{[(\log(n)+p)/n]^{0.4},0.05\}$. Our MM software shares this kernel and bandwidth and therefore optimizes the same objective. The response $y_i$ is generated as 
\begin{eqnarray}\label{eq:quantiley}
y_i & = & \tilde{\bx}_i^\top \bbeta^* + \Big(\frac{x_{ip}}{2} +1 \Big)[\epsilon_i - F^{-1}_{\epsilon_i}(q)],
\end{eqnarray}
where $\epsilon_i$ follows a $t$-distribution with 1.5 degrees of freedom and $F$ is the distribution function of the $t$-distribution. The term $- F^{-1}_{\epsilon_i}(q)$ creates quantile drift, which is a common practice in simulation studies of quantile regression. The multiplier $\frac{x_{ip}}{2} +1$ encourages  heteroskedasticity. We consider two quantile levels, $q=0.5$ and $q=0.8$, in our experiments. 

\begin{figure*}[htbp]
  \centering
  \includegraphics[width=\linewidth, height = 0.5\linewidth]{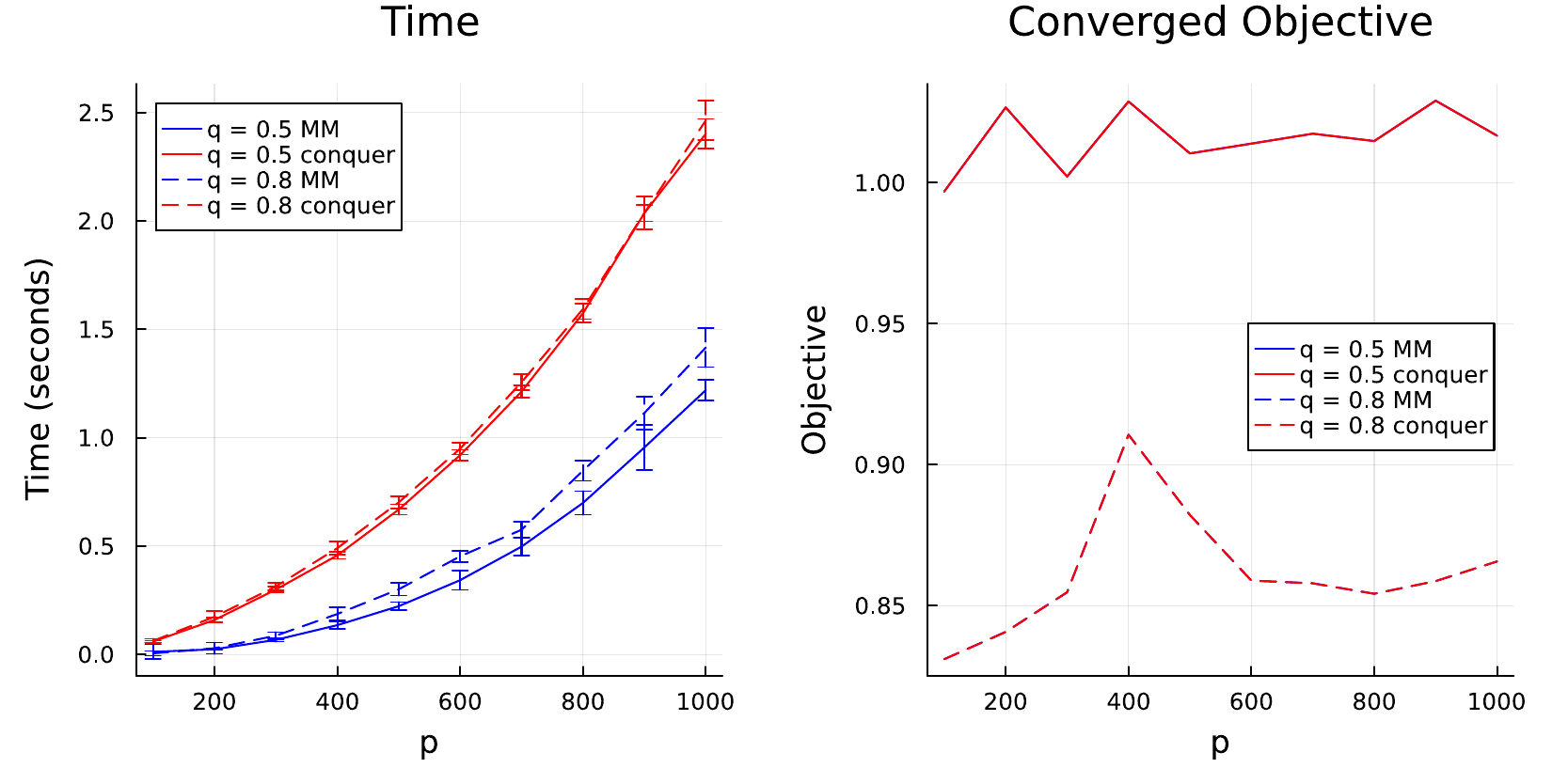}
  \caption{Results for non-sparse quantile regression. The left panel shows average computation time, with the error bars marking $\pm 1$ standard deviation. In the right panel, the objective trajectories of MM and \texttt{conquer} overlap.  }
  \label{quantile}
\end{figure*}

\textbf{$\text{L}_2\text{E}$ Regression} 
For $\text{L}_2\text{E}$ regression, we compare the proposed double majorization technique described in Section \ref{Sectionl2e} to iteratively reweighted least squares as proposed by \cite{liu2023sharper}.  Their implementation in R executes poorly on large-scale data. We re-implemented their IRLS approach in Julia and used it as our baseline competition. Responses are generated as 
\begin{eqnarray*}
y_i & = & \tilde{\bx}_i^\top \bbeta^* + \epsilon_i, \quad i = 1,2,\dots,n,
\end{eqnarray*}
where the $\epsilon_i$ are standard normal random deviates.   To induce outlier contamination, we add 10 to $y_i$ for the first 10\% of the samples and 10 to the second column of $\tilde{\bX}$ for the last 10\% of samples. This simulation choice demonstrates that the $\text{L}_2\text{E}$ estimator is also robust to contamination in the covariates $\bX$, a property not enjoyed by Huber regression or quantile regression.

\begin{figure*}[htbp]
  \centering
  \includegraphics[width=\linewidth, height = 0.5\linewidth]{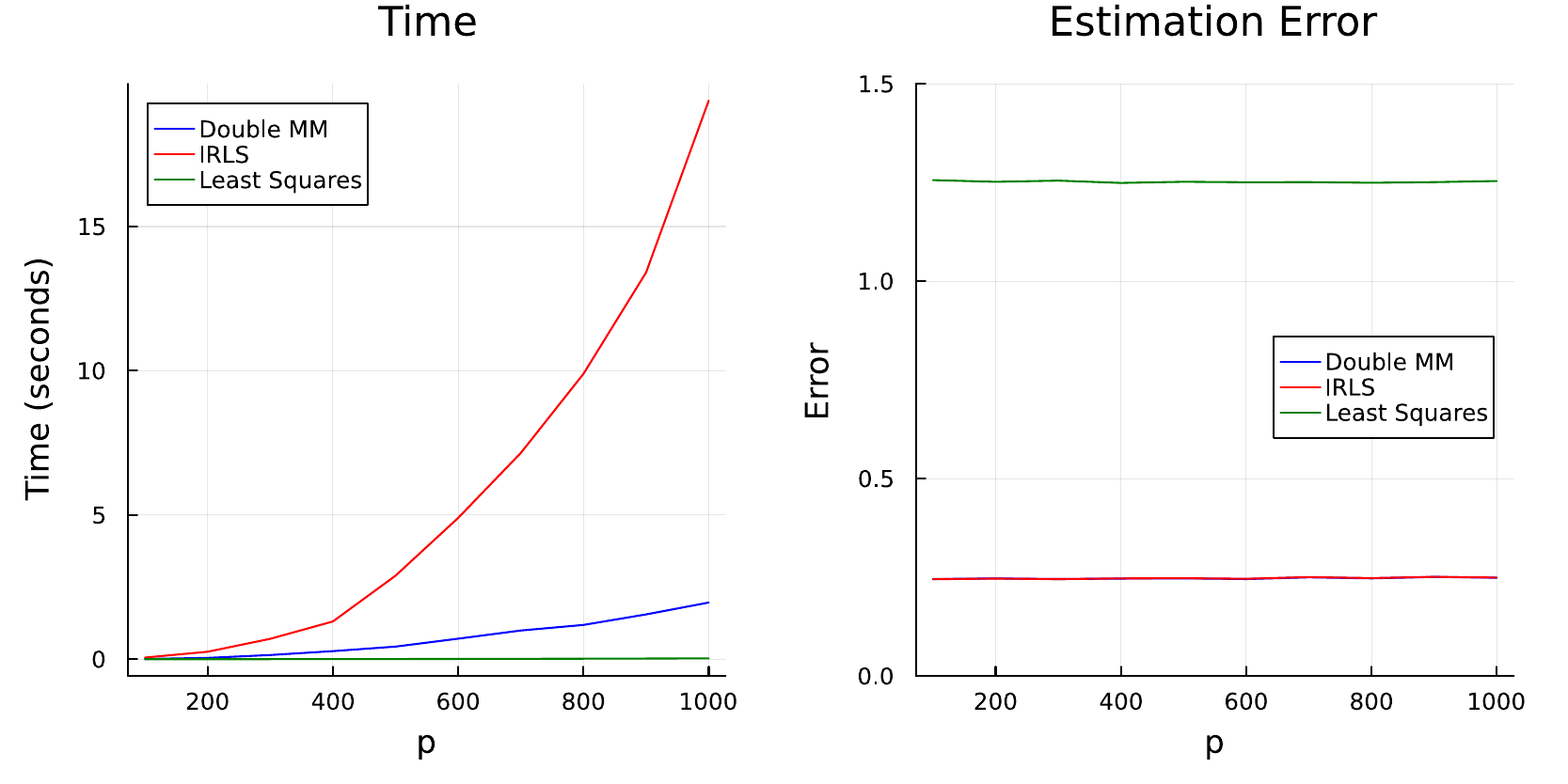}
  \caption{Results for $\text{L}_2\text{E}$ regression. IRLS and Least Squares refer to the weighted least squares approach of \cite{liu2023sharper} and naive least squares regression, respectively. In the right panel, the estimation error trajectories of Double MM and IRLS overlap.  }
  \label{l2e}
\end{figure*}

\textbf{Multinomial Regression} 
For multinomial (logistic) regression, we compare our MM algorithm to the \texttt{glmnet} R package \cite{friedman2010regularization} with penalty parameter 0. The responses $Y_i \in \Real^{c}$ are generated as
\begin{eqnarray*}
Y_i &\sim &\text{Multinomial}[1, (p_{i1}, p_{i2}, \dots, p_{ic})], \\ p_{ij} & = & \frac{\exp(\tilde{\bx}_i^\top \bbeta^*_j)}{\sum_{l=1}^c \exp(\tilde{\bx}_i^\top \bbeta^*_l)},
\end{eqnarray*}
where $\bbeta^*_j$ is the $j$-th column of the coefficient matrix $\bB^*\in \Real^{p\times c}$. The entries of $\bB^*$ are populated with random $\text{Uniform}[0,1]$ values.  In this simulation, there are $c=5$ categories. In contrast to the two previous  examples, $p$ ranges over $\{30,60,\dots,300\}$ and $n=1000p$. 

\begin{figure*}[htbp]
  \centering
  \includegraphics[width=\linewidth, height = 0.5\linewidth]{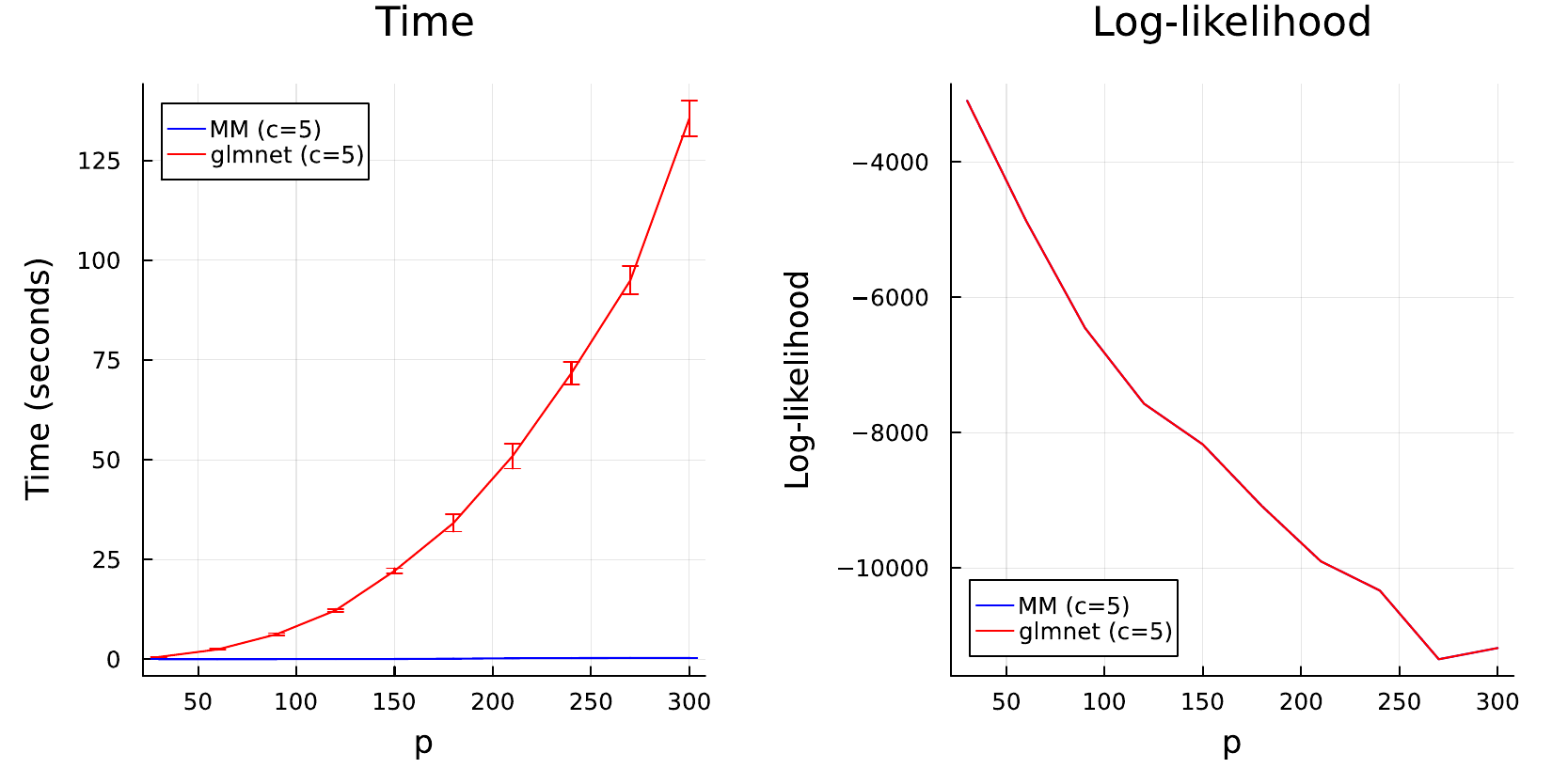}
  \caption{Results for multinomial regression. In the right panel, the log-likelihood trajectories of MM and \texttt{glmnet} overlap.}
  \label{multinomial}
\end{figure*}

Figures \ref{quantile}, \ref{l2e},  and \ref{multinomial} summarize our computational experiments. All experimental results are averaged over 50  random replicates. Figure \ref{quantile} shows that MM converges to the same objective as \texttt{conquer}, certifying the correctness of our implementation. However, as demonstrated in the left panel of Figure \ref{quantile}, MM can be 2 to 3 times faster than  \texttt{conquer}.  In $\text{L}_2\text{E}$ regression, Figure \ref{l2e} demonstrates that double majorization is over 10 times faster than weighted least squares. The right panel of  Figure \ref{l2e} shows that the estimation error $\| \hat{\bbeta} - \bbeta^* \|_2$ of the $\text{L}_2\text{E}$ estimator is much smaller than that of least squares, confirming the robustness of $\text{L}_2\text{E}$. Figure \ref{multinomial} shows that MM converges to the same log-likelihood as \texttt{glmnet}, but can be over 400 times faster for unpenalized multinomial regression. This speedup is attributable to the update rule \eqref{multinomialupdate}, which yields highly efficient updates by reducing the dimension of the normal equations. We note that unpenalized multinomial regression can be useful in constructing likelihood-ratio tests. Taken together, these results highlight MM's ability to compute common statistical estimators with high efficiency and no compromise in accuracy.

\subsection{Sparse Quantile Regression \label{models}}

We next compare sparse quantile regression with  $\ell_0$-norm and distance-to-set penalties to quantile regression with Lasso, SCAD, and MCP penalties. R package \texttt{conquer} supplies the Lasso, SCAD, and MCP results. In our comparisons, the design matrix $\tilde{\bX}$ is generated as described in Section \ref{sec:nonstructured}. However, following \cite{tan2022high} and \cite{man2024unified}, the true coefficient vector $\bbeta^*$ has $\bbeta^*_1 = 4$, $\bbeta^*_3 = 1.8$, $\bbeta^*_5=1.6$, $\bbeta^*_7=1.4$, $\bbeta^*_9=1.2$, $\bbeta^*_{11}=1$, $\bbeta^*_{13}= -1$, $\bbeta^*_{15}= -1.2$, $\bbeta^*_{17}= -1.4$, $\bbeta^*_{19}= -1.6$, and $\bbeta^*_{21}= -1.8$. All remaining elements of $\bbeta^*$ are 0. In other words, there are 10 nonzero predictors besides the intercept. The responses are generated by the protocol \eqref{eq:quantiley} except that we now consider $\mathcal{N}(0,2)$ noise in addition to the heavy-tailed $t_{1.5}$ noise. Our experiments cover both the over-determined case $(n=500,p=250)$ and the under-determined case $(n=250,p=500)$. Quantile levels are set at $q=0.5$ and $q=0.7$. We set the Moreau envelope parameter $\mu$ for the quantile loss to be the same as the bandwidth parameter $h$ returned by \texttt{conquer}, namely $\mu = \max\{0.05, \sqrt{q(1-q)} * (\log(p) / n)^{0.25}\}$.

The following metrics measure model performance: (a) true positive rate (TPR), (b) false positive rate (FPR), (c) estimation error (EE), defined as $\| \hat{\bbeta} - \bbeta^*\|_2$, and (d) prediction error (PE), defined as $\| \tilde{\bX}\hat{\bbeta} - \tilde{\bX}\bbeta^*\|_2$. Penalty parameters for all models are selected via 5-fold cross-validation, with quantile loss serving as the validation error. For the $\ell_0$-norm Moreau envelope penalty with penalty constant $\lambda$, we set $\alpha=0.01$ and search over an exponentially spaced grid of 30 different $\lambda$ points.  For the proximal distance model, we search over the grid $k \in \{1,2,\dots,30\}$ to identify the optimal sparsity level. We terminate the MM algorithms when the gradient norm of our loss function falls below $10^{-3}$.  While the optimal solution $\hat{\bbeta}$ is not sparse in the $\ell_0$-norm model, its proximal projection $\prox_{\alpha \|\cdot\|_0}(\hat{\bbeta})$ is sparse, and we use this projection as our final sparse estimator.
\begin{table*}[htbp]
\centering
\begin{tabular}{cccccc} 
\toprule
Method       & TPR           & FPR           & EE            & PE            & Time (s)  \\ 
\midrule
\multicolumn{6}{c}{$(n=500,p=250), q=0.5$}                                            \\
SQR-Lasso    & 1.00 (0.00)   & 0.09 (0.03)   & 0.46 (0.10)   & 8.07 (1.27)   & 2.39      \\
SQR-SCAD     & 0.96 (0.11)   & 0.02 (0.11)   & 2.42 (13.29)  & 30.5 (161.2)  & 2.96      \\
SQR-MCP      & 0.95 (0.11)   & 0.02 (0.12)   & 2.61 (13.58)  & 32.7 (166.2)  & 3.00      \\
SQR-$\ell_0$ & 1.00 (0.00)   & 0.00 (0.00)   & 0.24 (0.06)   & 4.36 (0.87)   & \bf{0.59}      \\
SQR-PD       & 1.00 (0.00)   & 0.00 (0.00)   & \bf{0.20 (0.06)}   & \bf{3.71 (1.12)}   & 9.90      \\ 
\midrule
\multicolumn{6}{c}{$(n=500,p=250), q=0.7$}                                            \\
SQR-Lasso    & 1.00 (0.00)   & 0.09 (0.03)   & 0.57 (0.13)   & 9.92 (2.14)   & 2.66      \\
SQR-SCAD     & 0.90 (0.13)   & 0.02 (0.12)   & 4.97 (26.93)  & 64.9 (348.5)  & 3.27      \\
SQR-MCP      & 0.91 (0.13)   & 0.02 (0.13)   & 4.99 (27.69)  & 65.3 (358.2)  & 3.36      \\
SQR-$\ell_0$ & 1.00 (0.03)   & 0.00 (0.00)   & 0.30 (0.18)   & 5.39 (2.70)   & \bf{0.68}      \\
SQR-PD       & 1.00 (0.01)   & 0.00 (0.00)   & \bf{0.29 (0.20)}   & \bf{5.36 (3.00)}   & 11.27     \\ 
\midrule
\multicolumn{6}{c}{$(n=250,p=500), q=0.5$}                                            \\
SQR-Lasso    & 1.00 (0.00)   & 0.06 (0.03)   & 0.86 (0.19)   & 10.7 (2.18)   & 2.89      \\
SQR-SCAD     & 0.73 (0.17)   & 0.10 (0.27)   & 4.33 (6.24)   & 52.1 (84.5)   & 3.55      \\
SQR-MCP      & 0.77 (0.14)   & 0.13 (0.30)   & 4.63 (6.47)   & 55.9 (87.7)   & 3.57      \\
SQR-$\ell_0$ & 0.99 (0.05)   & 0.00 (0.00)   & \bf{0.63 (0.38)}   & \bf{7.46 (3.71)}   & \bf{0.84}      \\
SQR-PD       & 0.99 (0.05)   & 0.01 (0.01)   & 0.66 (0.45)   & 8.50 (4.44)   & 33.03     \\ 
\midrule
\multicolumn{6}{c}{$(n=250,p=500), q=0.7$}                                            \\
SQR-Lasso    & 1.00 (0.01)   & 0.06 (0.02)   & 1.07 (0.20)   & 13.6 (2.51)   & 3.00      \\
SQR-SCAD     & 0.71 (0.15)   & 0.10 (0.27)   & 4.55 (6.06)   & 55.1 (82.8)   & 3.72      \\
SQR-MCP      & 0.77 (0.15)   & 0.18 (0.35)   & 4.96 (5.63)   & 60.9 (76.8)   & 3.71      \\
SQR-$\ell_0$ & 0.94 (0.09)   & 0.00 (0.00)   & \bf{0.97 (0.50)}   & \bf{11.2 (4.91)}   & \bf{0.89}      \\
SQR-PD       & 0.94 (0.14)   & 0.01 (0.01)   & 1.06 (0.98)   & 12.7 (9.22)   & 32.61     \\ 
\bottomrule
\end{tabular}
\caption{Simulation results for sparse quantile regression under heavy-tailed $t_{1.5}$ noise.}
\label{tab:sqrt}
\end{table*}

Table \ref{tab:sqrt} reports simulation results averaged over 50 random replicates under heavy-tailed noise. In general, Lasso-penalized quantile regression tends to produce an excess of false positives. SCAD and MCP penalties significantly reduce the number of false positives, but often at the expense of missing some true positives. In the under-determined case $(n=250, p=500)$ with heavy-tailed noise, SCAD and MCP perform poorly, missing many true positives while including many false positives. SCAD and MCP do perform better when the noise is normal, as the results tabulated  in Appendix \ref{sqrexperimentmore} indicate. The $\ell_0$-norm and distance-to-set penalties consistently achieve high TPR, low FPR, and low prediction error. The $\ell_0$-norm penalty delivers both high precision and efficient computation. Given its dense grid search and annealing of $\lambda$, distance-to-set estimation is generally slower than $\ell_0$-norm estimation.

\subsection{Robust Isotonic Regression}

To illustrate the application of robust isotonic regression, we consider global warming. The data set  \href{https://www.ncei.noaa.gov/access/monitoring/climate-at-a-glance/global/time-series}{climate at a glance}  records December land and ocean temperature anomalies from 1850 to 2023 relative to the 1901-2000 mean. This data set is similar to one analyzed by \cite{wu2001isotonic} and \cite{tibshirani2011nearly}. Appendix \ref{robust_isotonic_appendix} records the implementation details of our 
$\text{L}_2\text{E}$ method. The method not only enables robust estimation, but also quantifies the outlyingness of each observation through a converged case weight $w_i$.  We compare the method, which recycles matrix decompositions, to the weighted least squares method of \cite{liu2023sharper}.  As a non-robust benchmark, we include the non-robust isotonic fit delivered by the pool adjacent violators algorithm \cite{barlow1972isotonic}. 

\begin{figure*}[htbp]
  \centering
  \includegraphics[width=\linewidth, height = 0.45\linewidth]{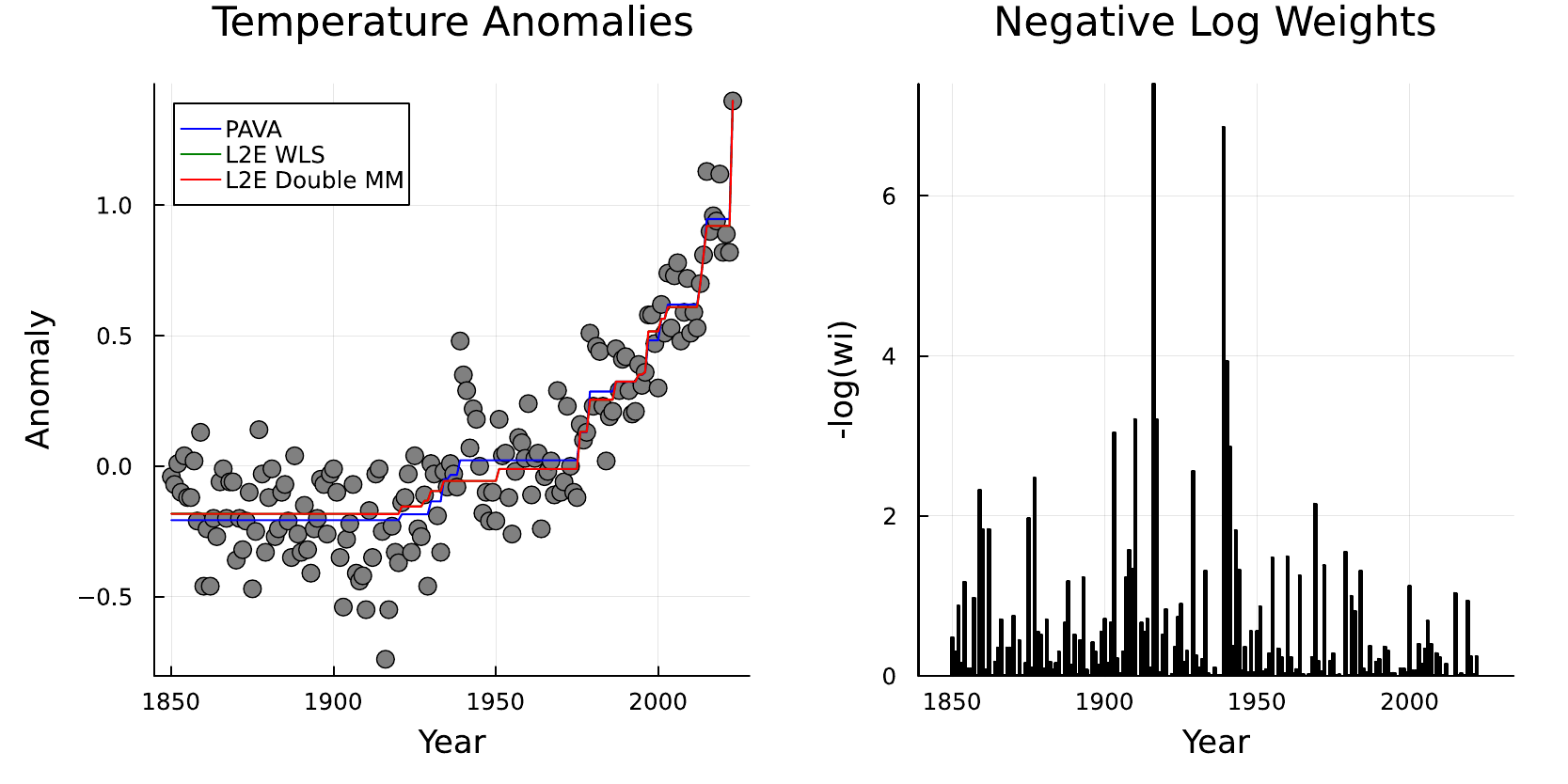}
  \caption{Results for robust isotonic regression. In the left panel, the green and red lines of the robust fits overlap, rendering the green line invisible.}
  \label{fig:L2EIsotonic}
\end{figure*}

The two robust methods give almost identical fits, as depicted in the left panel of Figure \ref{fig:L2EIsotonic}. However, our method takes 0.17 seconds, while IRLS takes 1.84 seconds. These results echo the findings in Section \ref{sec:nonstructured} about efficient computation. Compared to the non-robust fit, our robust isotonic estimate appears to be less influenced by the unusually low temperatures near 1920 and the unusually high ones near 1940. The right panel of Figure \ref{fig:L2EIsotonic} depicts the values  $-\log(w_i)$ quantifying the outlyingness of each observation.
The two spikes near 1920 and 1940 are clearly visible. 

\subsection{Low-Rank Multinomial Regression}

Simulated data from smoothed-nuclear-norm penalized multinomial regression allow us to compare with the R package \texttt{npmr} \cite{powers2018nuclear}, which is powered by accelerated proximal gradient descent. Section \ref{quad_bound_section} described the details of our multinomial model under the smoothed-nuclear-norm penalty. The true coefficient matrix $\bB^*$ is populated with $\operatorname{Uniform}[0,3]$ entries and then projected to the closest rank-3 subspace, with the intercept row excluded from projection to avoid penalizing intercepts. The design matrix and responses are generated according to the protocols of Section~\ref{sec:nonstructured}. We consider two settings, $(n,p,c)=(1000,100,10)$ and $(n,p,c)=(5000,250,20)$, and carry out a grid search on the penalty constant $\lambda$ over 50 points decreasing exponentially from $1$ to $10^{-4}$. For each of 20 replicates, we generate three independent data sets of equal size: a training set to fit the model, a validation set to select the optimal $\lambda$, and a test set to evaluate out-of-sample performance via test log-likelihood. We set the Moreau smoothing constant to $\mu=0.1$ and terminate the MM algorithms when $\|\nabla f(\bB)\| < 10^{-3}$.

Table~\ref{tab:lowrank} summarizes the test log-likelihood (using optimal $\lambda$ tuned on the validation set) and computation time for the full solution path. Both methods achieve comparable test log-likelihoods, although a small discrepancy exists between our MM algorithm and \texttt{npmr}, likely because \texttt{npmr} adopts a slightly different multinomial parameterization with $c$ columns in the coefficient matrix instead of $c-1$. Our smoothed-nuclear-norm MM algorithm is approximately 25 times faster in the first setting and over 10 times faster in the second setting. Several factors contribute to this speed advantage. First, although proximal gradient descent benefits from backtracking line search, the choice of initial step size can substantially affect convergence speed and typically requires careful tuning; our experiments simply adopt the default of \texttt{npmr}. In contrast, the MM algorithm requires no step size specification, as majorization automatically guarantees monotone descent. Second, our MM algorithm incorporates curvature information through the quadratic upper bound majorization, yielding much faster convergence. Finally, our implementation avoids repeated matrix decompositions by reducing the normal equations to a Sylvester equation, a key tactic for reducing computational complexity.
\begin{table}
\centering
\caption{Comparison of MM and \texttt{npmr} on simulated data.}
\label{tab:lowrank}
\begin{tabular}{lcc}
\toprule
Method & Test log-likelihood & Time (s) \\
\midrule
\multicolumn{3}{l}{\textit{$n=1000, p=100, c=10$}} \\
\midrule
MM ($\mu=0.1$) & $-1019.8 \pm 104.5$ & $0.77 \pm 0.07$ \\
\texttt{npmr}                    & $-1001.8 \pm 103.8$ & $18.67 \pm 7.59$ \\
\midrule
\multicolumn{3}{l}{\textit{$n=5000, p=250, c=20$}} \\
\midrule
MM ($\mu=0.1$) & $-5381.9 \pm 639.2$ & $13.81 \pm 2.01$ \\
\texttt{npmr}                    & $-5315.9 \pm 565.3$ & $147.56 \pm 63.48$ \\
\bottomrule
\end{tabular}
\end{table}

We also applied MM-based low-rank multinomial regression to the \texttt{Vowel} data set available in the \texttt{npmr} package. The 990 observations on 10 independent predictors record speaker-independent recognition of the 11 steady-state vowels of British English. After adding an intercept, $(p, c) = (11, 11)$. The 10 predictors are log-area ratios under linear predictive coding. The data are divided into a training set of 528 observations and a test set of 462 observations. We compute the full solution path on the training set with $\lambda$ ranging over 50 exponentially spaced points between $10^{-4}$ and $1$. The test set is randomly split into a validation half and a test half: the validation half selects the optimal $\lambda$, and the test half evaluates out-of-sample performance. Results are averaged over 20 random splits. We report the mean test log-likelihood and total computation time for MM-based low-rank multinomial regression ($\mu = 0.1$), \texttt{npmr} low-rank multinomial regression, and ordinary multinomial regression without low-rank regularization. 
\begin{table}
\centering
\caption{Results on the Vowel Data.}
\label{tab:vowel}
\begin{tabular}{ccc}
\toprule
Method & Test log-likelihood & Time (s) \\
\midrule
MM ($\mu=0.1$) & $-293.5 \pm 7.4$  & $0.21 \pm 0.01$ \\
\texttt{npmr}                    & $-292.9 \pm 7.9$  & $6.66 \pm 0.11$ \\
Multinomial           & $-593.1 \pm 36.4$ & $0.03 \pm 0.08$ \\
\bottomrule
\end{tabular}
\end{table}
Table \ref{tab:vowel} summarizes the experimental results. The table shows that regularized multinomial regression yields much higher test log-likelihoods than ordinary multinomial regression, demonstrating the benefits of low-rank regularization. Despite its reliance on nuclear norm smoothing, our MM algorithm also yields test log-likelihoods similar to those of \texttt{npmr}. Our code executes over 30 times faster than \texttt{npmr}, reinforcing our previous claims about computational efficiency.

\subsection{A Comparison of Acceleration Strategies}\label{ablation}
\begin{figure*}[tbp]
  \centering
  \includegraphics[width=\linewidth, height =0.5\linewidth]{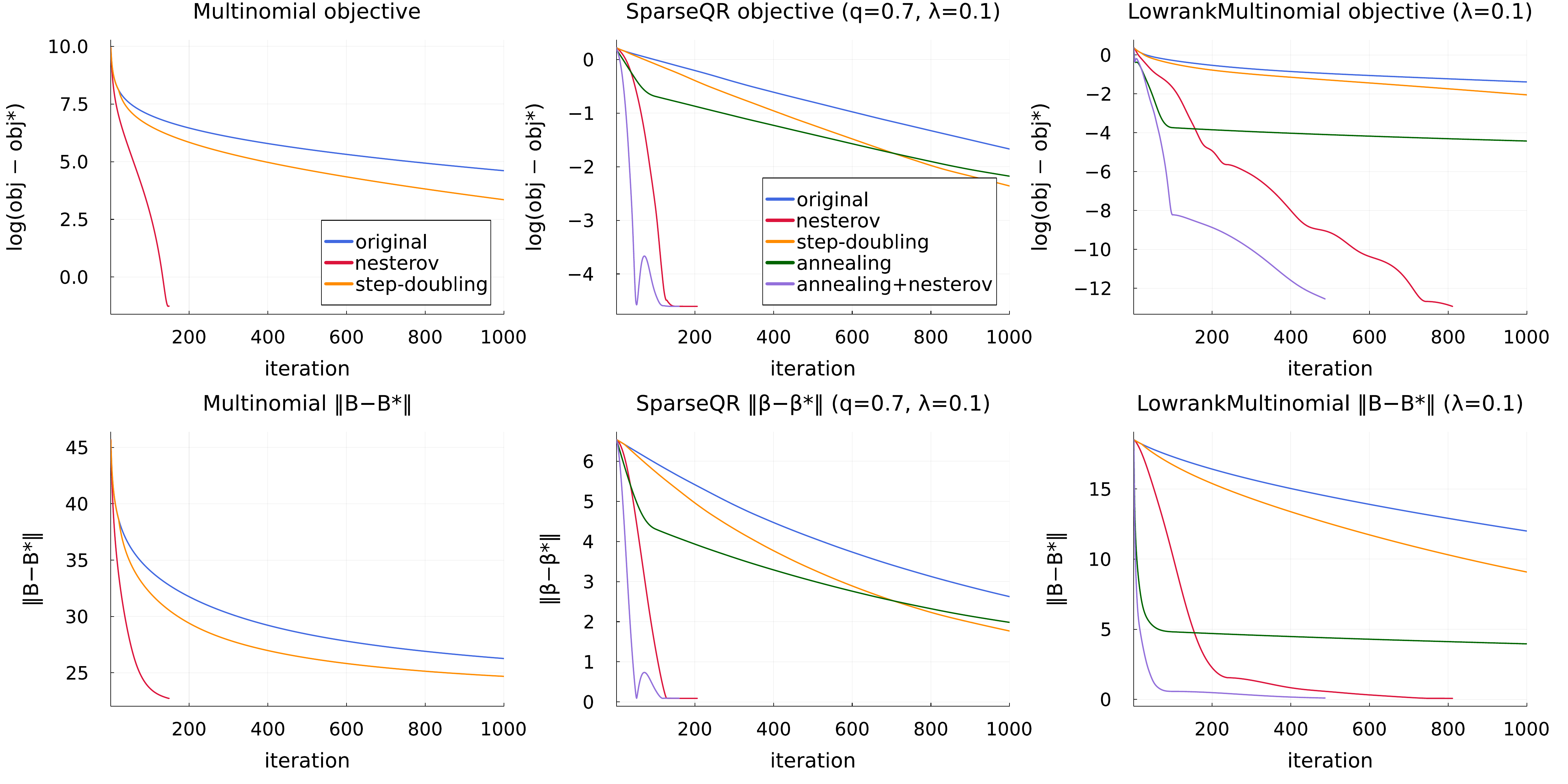}
  \caption{Visualization of the effects of different acceleration strategies on multinomial regression, sparse quantile regression (using smoothed $\ell_0$ norm), and low-rank multinomial regression. $\text{obj}^*$, $\bbeta^*$ and $\bB^*$ are obtained by running the algorithm with Nesterov acceleration for 10000 iterations. }
  \label{fig:comparison}
\end{figure*}
In this subsection, we conduct experiments to evaluate the effects of restarted Nesterov acceleration, step-doubling, and smoothing-parameter annealing. We consider three problems: unpenalized multinomial regression $(n = 10000,\; p = 100)$, sparse quantile regression under the smoothed $\ell_0$ penalty $(n = 500,\; p = 1000)$, and low-rank multinomial regression $(n = 5000,\; p = 250,\; c = 20)$. The data generation protocol is the same as previously described. We terminate the MM algorithms when $\|\nabla f(\bbeta)\| < 10^{-3}$.

For multinomial regression, we compare three variants: no acceleration, step-doubling, and Nesterov acceleration with restart. In step-doubling, we skip the first 20 iterations to avoid being stuck at large function values.  For the latter two problems, we also consider annealing $\mu$ from a large initial value. Specifically, $\mu$ decays exponentially from $\mu_{\max} = 10.0$ to $\mu_{\min} = 0.01$ over 100 iterations. We also consider the combination of annealing with Nesterov acceleration. Figure~\ref{fig:comparison} aids visualization of two types of trajectories: the objective value (evaluated at $\mu = \mu_{\min}$), and the distance to the final converged solution ($\bbeta^*$ or $\bB^*$), obtained by running the algorithm for 10,000 iterations with Nesterov acceleration.

In general, restarted Nesterov acceleration yields a more substantial speedup than step-doubling or annealing. The annealing strategy encourages large steps in early iterations, quickly bringing the iterates into the vicinity of the optimal solution. A combination of annealing and Nesterov acceleration appears to be the most effective strategy when a small smoothing parameter is desired. This synergy leverages both the rapid initial progress of annealing and the sustained impact of restarted Nesterov acceleration.

\subsection{Effect of the Smoothing Parameter on Estimation Error}\label{impact}
The Moreau envelope provides a smooth approximation to the original loss and penalty functions. A larger $\mu$ yields a coarser approximation but reduces computational difficulty. We now investigate how the smoothing parameter $\mu$ affects the estimation error of our estimators.

We consider four settings: Huber regression with varying $\mu$ $(n = 10000,\; p = 100)$; sparse quantile regression under the smoothed $\ell_0$ penalty $(n = 250,\; p = 500,\; q = 0.7)$ with either varying $\mu$ or varying $\alpha$; and low-rank multinomial regression $(n = 2000,\; p = 50,\; c = 5)$ with varying $\mu$. The smoothing parameter ranges over 30 exponentially spaced points between $10^2$ and $10^{-3}$. When we vary $\mu$ for sparse quantile regression, we fix $\alpha$ at $0.01$. Likewise, when we vary $\alpha$ for sparse quantile regression, we fix $\mu$ at $0.25$. The data generation protocol remains as previously described. For sparse quantile regression and low-rank multinomial regression, we generate an independent validation set and select the optimal $\lambda$ from a grid of 30 points decreasing exponentially from $10$ to $10^{-4}$, using quantile loss or negative log-likelihood on the validation set. We report $\|\hat{\bbeta} - \bbeta_{\text{true}}\|$ or $\|\hat{\bB} - \bB_{\text{true}}\|$ as the final estimation error, where $\bbeta_{\text{true}}$ and $\bB_{\text{true}}$ are the true coefficients used to generate the data.

\begin{figure*}[tbp]
  \centering
  \includegraphics[width=\linewidth, height =0.5\linewidth]{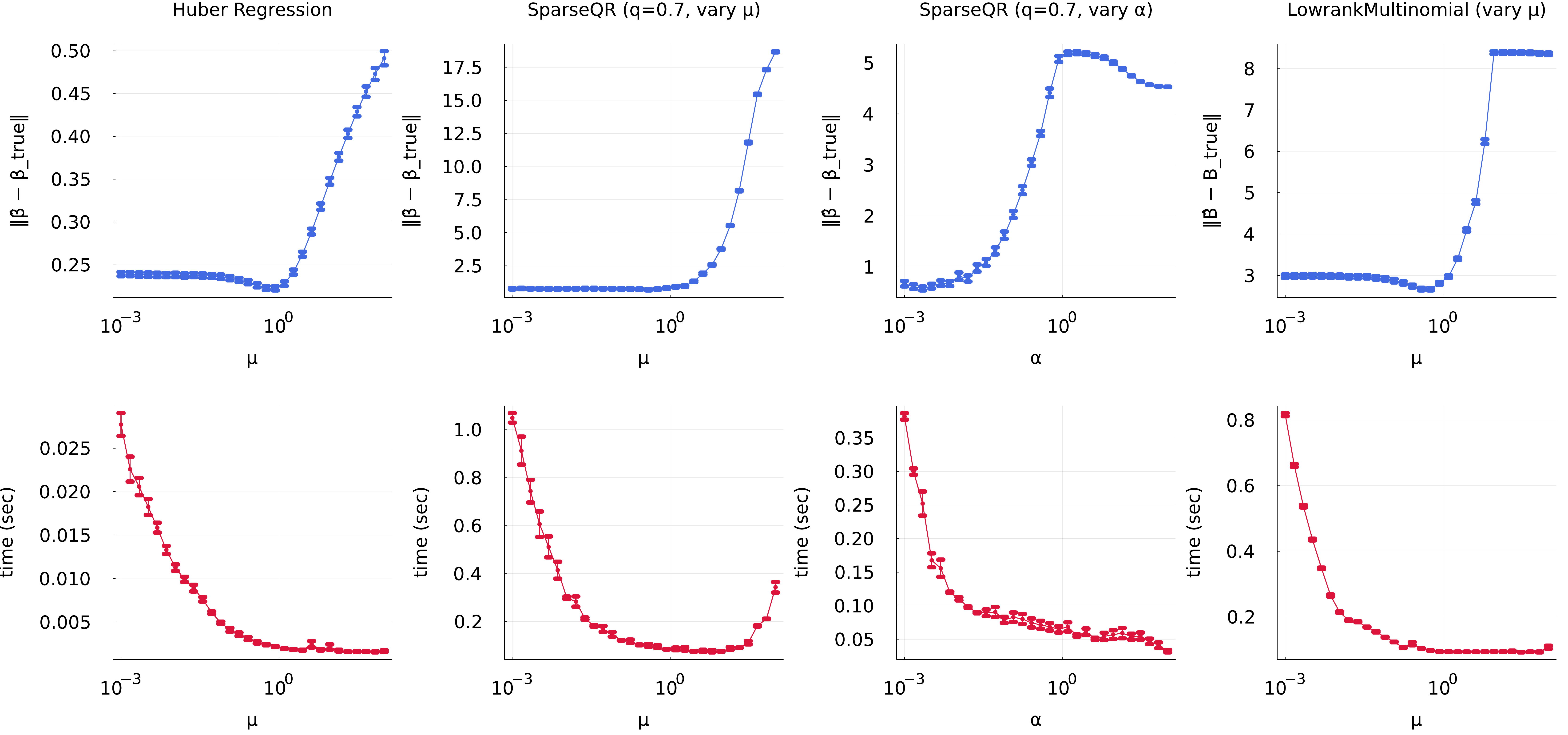}
  \caption{Estimation error and computation time for several problems under varying smoothing parameters. Data points are averaged over 50 random replicates, and the error bar denotes $\pm 1$ standard deviation. }
  \label{fig:impact}
\end{figure*}

Figure~\ref{fig:impact} shows that an overly large $\mu$ or $\alpha$ leads to gross estimation errors. In general, estimation error decreases as $\mu$ and $\alpha$ decrease. Interestingly, however, the smallest smoothing constant does not always yield the most accurate estimates. For Huber regression, the lowest estimation error occurs at a small but nonnegligible value of $\mu$ (near 1), suggesting the existence of an optimal smoothing level. This phenomenon has been highlighted in recent papers on adaptive Huber regression \cite{sun2020adaptive} and smoothed quantile regression \cite{fernandes2021smoothing} and can be understood as an
example of the bias–variance tradeoff. Although smoothing the check loss introduces bias in estimating regression coefficients, it reduces variance and improves statistical efficiency, so the resulting estimators may achieve lower mean squared error in finite samples. 

For sparse quantile regression, estimation error (under the best $\lambda$) appears to plateau once the smoothing constant becomes sufficiently small. For low-rank multinomial regression, there also appears to be an optimal smoothing level. As the second row of Figure~\ref{fig:impact} shows, the computation time of our MM algorithms generally decreases as the smoothing constant grows. These observations suggest a simple practical guideline. Rather than driving $\mu$ and $\alpha$ to near zero, choosing a small but nonnegligible smoothing constant simplifies optimization and yields estimation accuracy competitive with its nonsmooth counterparts.

\section{Discussion}

The current paper stresses the MM principle, smoothing by Moreau envelopes, and the recycling of matrix decompositions. The MM principle simplifies optimization by: (a) separating the variables of a problem, (b) avoiding large matrix inversions, (c) linearizing a problem, (d) restoring symmetry, (e) dealing with equality and inequality constraints gracefully, and (f) turning a nondifferentiable problem into a smooth problem. Our examples feature quadratic surrogates, by far the most common surrogates and the easiest to implement in practice. Quadratic surrogates are not necessarily parameter separated, and their stationary conditions require solving linear equations. This is where recycling of matrix decompositions comes into play. Cholesky, QR, and spectral decompositions are expensive to extract in high dimensions, so there is strong motivation to employ quadratic surrogates with the same curvature across all iterations. Deweighting and Moreau majorization are two tactics achieving this goal. Although Cholesky decompositions are more likely to suffer from ill conditioning than QR decompositions, they are faster to compute and feature exclusively in our numerical examples. 

Many loss and penalty functions are nondifferentiable. Smoothing by Moreau envelopes and other approximation methods removes the kinks and allows the standard rules of calculus to apply in optimization. A Moreau envelope can be majorized by a spherical quadratic with parameters separated. This advantage extends to squared distance-to-set penalties in constrained optimization. These majorizations are not particularly tight, so it often takes many iterations to converge. Fortunately, Nesterov acceleration mitigates most of the drawbacks. Our
numerical examples illustrate the virtues of these tactics. Computational speed is often an order of magnitude better than that delivered by popular competing methods. More subtly, statistical inference is improved by the parsimony imposed by regularization and by outlier flagging. 

In any event, it is the combination of tactics that produces the fastest and most reliable algorithms. Techniques such as alternating minimization (block descent and ascent), profile log-likelihoods, penalty constant annealing \cite{zhou2010bumpy}, and quasi-Newton methods \cite{zhou2011quasi} have also proved valuable in many settings. Readers should keep in mind that the MM principle can be invoked on the subproblems of alternating minimization. Progress, not perfection, is usually adequate in solving the subproblems. Optimization requires both art and science. Although there is no panacea, overarching themes can guide the construction of good algorithms. Our examples vividly demonstrate the speed improvements achievable through thoughtful algorithm design. There are doubtless many advances yet to be made in accelerating the optimization algorithms so vital to the progress of statistics. Readers who want to replicate our experiments and extend our Julia code can visit the website \url{https://github.com/qhengncsu/MMDeweighting.jl}.

\begin{funding}
This research was partially funded by grants from the National Institutes of Health (R35GM141798, HZ and KL; R01DK142026, HZ) and the National Science Foundation (DMS-2054253 and IIS-2205441, HZ).
\end{funding}

\begin{appendix}

\section{Deweighting the Sharpest LAD Majorization}\label{double_major}

In LAD regression, deweighting the sharpest quadratic majorization \eqref{bestquadratic} 
leads to the surrogate \eqref{ladsurrogate} derived from the Moreau majorization 
\eqref{Moreau_majorization}. To prove this assertion, recall that the Moreau surrogate 
\eqref{ladsurrogate} is
\begin{equation*}
    g(\bbeta \mid \bbeta_m) = \frac{1}{2\mu}\sum_{i=1}^n (r_i - z_{mi})^2,
\end{equation*}
where $r_i = y_i - \bx_i^\top\bbeta$ and $z_{mi} = \prox_{\mu|\cdot|}(r_{mi})$. 
On the other hand, the best quadratic majorization \eqref{bestquadratic} in LAD is
\begin{equation*}
    g_1(\bbeta \mid \bbeta_m) = \frac{1}{2\mu}\sum_{i=1}^n w_{mi} r_i^2 + c_m,
\end{equation*}
where $w_{mi} = 1$ if $|r_{mi}| < \mu$ and $w_{mi} = \mu/|r_{mi}|$ if $|r_{mi}| \ge \mu$. 
In both cases the weights satisfy $0 \le w_{mi} \le 1$. Deweighting 
$g_1(\bbeta \mid \bbeta_m)$ yields the new surrogate
\begin{equation*}
    g_2(\bbeta \mid \bbeta_m) = \frac{1}{2\mu}\sum_{i=1}^n 
    \bigl[w_{mi} y_i + (1 - w_{mi})\bx_i^\top\bbeta_m - \bx_i^\top\bbeta\bigr]^2.
\end{equation*}
When $|r_{mi}| \le \mu$,
\begin{equation*}
    w_{mi} y_i + (1 - w_{mi})\bx_i^\top\bbeta_m 
    \;=\; y_i 
    \;=\; y_i - \prox_{\mu|\cdot|}(r_{mi}),
\end{equation*}
and when $|r_{mi}| > \mu$,
\begin{align*}
    w_{mi} y_i + (1 - w_{mi})\bx_i^\top\bbeta_m 
    &= y_i - (1 - w_{mi})(y_i - \bx_i^\top\bbeta_m) \\
    &= y_i - (1 - w_{mi})\, r_{mi} \\
    &= y_i - (r_{mi} - \mu) \\
    &= y_i - \prox_{\mu|\cdot|}(r_{mi}).
\end{align*}
In both cases, $w_{mi}y_i+(1-w_{mi})\bx_i^\top\bbeta_m = y_i - z_{mi}$. 
It follows that $g_2(\bbeta\mid\bbeta_m)=g(\bbeta\mid\bbeta_m)$, 
completing the proof.
\section{Robust Isotonic Regression} \label{robust_isotonic_appendix}
The $\text{L}_2\text{E}$ version of robust isotonic regression is driven by the 
penalized loss
\begin{align}\label{eq:robustisotonic}
    f(\bbeta, \tau) &= \frac{\tau}{2\sqrt{\pi}} - \frac{\tau}{n}\sqrt{\frac{2}{\pi}} 
    \sum_{i=1}^{n} e^{-\frac{\tau^{2}}{2}(y_i - \beta_i)^2} \notag \\
    &\quad + \frac{\lambda}{2} \operatorname{dist}^2\!\bigl(\bD\bbeta,\, 
    \mathbb{R}^{p-1}_+\bigr),
\end{align}
where $\bD \in \Real^{(p-1)\times p}$ is a matrix that generates the differences of adjacent components of 
$\bbeta$, and $\mathbb{R}^{p-1}_+$ is the $(p-1)$-dimensional nonnegative orthant. 
Our double majorization of the $\text{L}_2\text{E}$ loss and distance majorization 
of the set penalty together produces the MM surrogate
\begin{align}\label{eq:suristonic}
    g(\bbeta \mid \bbeta_m) &= \frac{\tau^3}{2n}\sqrt{\frac{2}{\pi}} \sum_{i=1}^n 
    \bigl[w_{mi} y_i + (1 - w_{mi})\beta_{mi} - \beta_i\bigr]^2 \notag \\
    &\quad + \frac{\lambda}{2}\bigl\| \bD\bbeta - 
    P_{\mathbb{R}^{p-1}_+}(\bD\bbeta_m) \bigr\|_2^2 + c_m.
\end{align}
The stationary condition for updating $\bbeta$ is
\begin{align}\label{eq:stationaryisotonic}
    &\left(\frac{\tau^3}{n}\sqrt{\frac{2}{\pi}}\, \bI_p 
    + \lambda \bD^\top \bD\right) \bbeta \notag \\
    &\quad = \frac{\tau^3}{n}\sqrt{\frac{2}{\pi}}\, \tilde{\by} 
    + \lambda\, \bD^\top P_{\mathbb{R}^{p-1}_+}(\bD\bbeta_m),
\end{align}
where $\tilde{\by}$ is the vector of shifted responses with components 
$\tilde{y}_i = w_{mi} y_i + (1 - w_{mi})\beta_{mi}$. The spectral decomposition of $\bD^\top\bD$ is computed once and recycled across all iterations and all values of $\tau$ and $\lambda$.
Block descent alternates 
the updates of $\tau$ and $\bbeta$. The precision parameter $\tau$ can be updated 
by gradient descent \cite{heng2023robust} or an approximate Newton's method 
\cite{liu2023sharper}.

\section{Sparse Quantile Regression under Gaussian Noise}
\label{sqrexperimentmore}
Table \ref{tab:sqrn} records the simulation results for sparse quantile regression under $\mathcal{N}(0,2)$ noise.
\begin{table*}[htbp]
\centering
\begin{tabular}{cccccc} 
\toprule
Method                         & TPR         & FPR         & EE          & PE           & Time (s)  \\ 
\midrule
\multicolumn{6}{c}{$(n=500,p=250), q=0.5$}                                                       \\
SQR-Lasso                      & 1.00 (0.00) & 0.09 (0.04) & 0.48 (0.09) & 8.64 (1.34)  & 1.43      \\
SQR-SCAD                       & 1.00 (0.03) & 0.00 (0.00) & 0.24 (0.19) & 4.46 (2.56)  & 1.32      \\
SQR-MCP                        & 1.00 (0.00) & 0.00 (0.00) & \bf{0.22 (0.06)} & \bf{4.13 (0.94)}  & 1.39      \\
SQR-$\ell_0$                   & 1.00 (0.00) & 0.00 (0.00) & 0.26 (0.07) & 4.99 (1.17)  & \bf{0.52}      \\
SQR-PD                         & 1.00 (0.00) & 0.00 (0.00) & 0.24 (0.07) & 4.58 (1.30)  & 9.75      \\ 
\midrule
\multicolumn{6}{c}{$(n=500,p=250), q=0.7$}                                                       \\
SQR-Lasso                      & 1.00 (0.00) & 0.11 (0.04) & 0.51 (0.09) & 9.26 (1.44)  & 1.54      \\
SQR-SCAD                       & 1.00 (0.00) & 0.00 (0.00) & \bf{0.24 (0.07)} & \bf{4.38 (1.11)}  & 1.42      \\
SQR-MCP                        & 1.00 (0.03) & 0.00 (0.00) & 0.26 (0.19) & 4.74 (2.72)  & 1.47      \\
SQR-$\ell_0$                   & 1.00 (0.00) & 0.00 (0.00) & 0.28 (0.08) & 5.28 (1.37)  & \bf{0.58}      \\
SQR-PD                         & 1.00 (0.00) & 0.00 (0.00) & 0.26 (0.09) & 4.75 (1.63)  & 10.22     \\ 
\midrule
\multicolumn{6}{c}{$(n=250,p=500), q=0.5$}                                                       \\
SQR-Lasso                      & 1.00 (0.00) & 0.06 (0.02) & 0.90 (0.17) & 11.25 (1.90) & 2.14      \\
SQR-SCAD                       & 0.93 (0.09) & 0.00 (0.00) & 0.79 (0.57) & 8.79 (5.44)  & 1.85      \\
SQR-MCP                        & 0.91 (0.10) & 0.00 (0.00) & 0.90 (0.58) & 9.94 (5.68)  & 1.88      \\
SQR-$\ell_0$                   & 0.99 (0.04) & 0.00 (0.00) & \bf{0.56 (0.26)} & \bf{6.75 (2.42)}  & \bf{0.64}      \\
SQR-PD                         & 0.97 (0.09) & 0.01 (0.01) & 0.84 (0.79) & 9.39 (6.00)  & 25.87     \\ 
\midrule
\multicolumn{6}{c}{$(n=250,p=500), q=0.7$}                                                       \\
SQR-Lasso                      & 1.00 (0.00) & 0.06 (0.02) & 0.94 (0.19) & 11.82 (2.26) & 2.25      \\
SQR-SCAD                       & 0.91 (0.10) & 0.00 (0.00) & 0.93 (0.59) & 10.47 (5.91) & 1.92      \\
SQR-MCP                        & 0.92 (0.09) & 0.00 (0.00) & 0.88 (0.58) & 9.71 (5.64)  & 1.99      \\
SQR-$\ell_0$                   & 0.98 (0.06) & 0.00 (0.00) & \bf{0.67 (0.36)} & \bf{8.56 (3.25)}  & \bf{0.77}      \\
SQR-PD                         & 0.98 (0.07) & 0.01 (0.01) & 0.79 (0.60) & 9.67 (5.19)  & 26.80     \\
\bottomrule
\end{tabular}
\caption{Simulation results for sparse quantile regression under $\mathcal{N}(0,2)$ noise.}
\label{tab:sqrn}
\end{table*}
\section{Proofs \label{proofs_appendix}}
\subsection{Proof of Proposition \ref{thm:convexconverge}}

\begin{proof}
Our attack exploits the inner product $\langle \balpha, \bbeta \rangle_{\bH} = 
\balpha^\top \bH \bbeta$ and corresponding norm $\|\bbeta\|_{\bH} = 
\sqrt{\bbeta^\top \bH \bbeta}$ associated with the positive definite second 
differential $\bH = d^2g(\bbeta \mid \bbeta)$. Taking 
$\balpha = \bbeta - \bH^{-1}\nabla f(\bbeta)$ in the majorization
\begin{align}
    f(\balpha) &\le f(\bbeta) + \nabla f(\bbeta)^\top (\balpha - \bbeta) \notag \\
    &\quad + \frac{1}{2}(\balpha - \bbeta)^\top \bH(\balpha - \bbeta) 
    \label{Hmajorization}
\end{align}
leads to the conclusion
\begin{align}\label{eq:kl}
    \inf_{\balpha}\, f(\balpha) &\le f(\bbeta) 
    - \frac{1}{2} \nabla f(\bbeta)^\top \bH^{-1} \nabla f(\bbeta).
\end{align}
Now consider the function $g_{\bbeta}(\balpha) = f(\balpha) - \nabla f(\bbeta)^\top \balpha$. 
It is convex, achieves its minimum at the stationary point $\balpha = \bbeta$, and 
satisfies the analogues of inequalities \eqref{Hmajorization} and \eqref{eq:kl}.
Therefore,
\begin{align*}
    &f(\balpha) - f(\bbeta) - \nabla f(\bbeta)^\top (\balpha - \bbeta) \\
    &\quad = g_{\bbeta}(\balpha) - g_{\bbeta}(\bbeta) \\
    &\quad \ge \frac{1}{2} \nabla g_{\bbeta}(\balpha)^\top \bH^{-1} 
    \nabla g_{\bbeta}(\balpha) \\
    &\quad = \frac{1}{2} [\nabla f(\bbeta) - \nabla f(\balpha)]^\top 
    \bH^{-1}[\nabla f(\bbeta) - \nabla f(\balpha)].
\end{align*}
By symmetry,
\begin{align*}
    &f(\bbeta) - f(\balpha) - \nabla f(\balpha)^\top (\bbeta - \balpha) \\
    &\quad \ge \frac{1}{2} [\nabla f(\bbeta) - \nabla f(\balpha)]^\top 
    \bH^{-1}[\nabla f(\bbeta) - \nabla f(\balpha)],
\end{align*}
and adding the last two inequalities gives
\begin{align*}
    &[\nabla f(\bbeta) - \nabla f(\balpha)]^\top 
    \bH^{-1}[\nabla f(\bbeta) - \nabla f(\balpha)] \\
    &\quad \le [\nabla f(\bbeta) - \nabla f(\balpha)]^\top (\bbeta - \balpha).
\end{align*}
We are now in a position to prove that the operator 
$S(\bbeta) = \bbeta - 2\bH^{-1} \nabla f(\bbeta)$ is non-expansive in the 
norm $\|\cdot\|_{\bH}$. Indeed,
\begin{align*}
    &\|S(\bbeta) - S(\balpha)\|_{\bH}^2 \\
    &\quad = \|\bbeta - \balpha\|^2_{\bH} 
    - 4(\bbeta - \balpha)^\top [\nabla f(\bbeta) - \nabla f(\balpha)] \\
    &\qquad + 4[\nabla f(\bbeta) - \nabla f(\balpha)]^\top 
    \bH^{-1}[\nabla f(\bbeta) - \nabla f(\balpha)] \\
    &\quad \le \|\bbeta - \balpha\|^2_{\bH}.
\end{align*}
The related operator
\begin{align*}
    T(\bbeta) &= \frac{1}{2}\bbeta + \frac{1}{2}S(\bbeta) 
    = \bbeta - \bH^{-1}\nabla f(\bbeta)
\end{align*}
is $\frac{1}{2}$-averaged with fixed points equal to the stationary points of 
$f(\bbeta)$. To complete the proof, we note that the sequence 
$\bbeta_{m+1} = T(\bbeta_m)$ defined by an averaged operator is known to converge 
to a fixed point; see, for example, Proposition~7.5.2 of \cite{lange2016mm} or 
Corollary~22.20 of \cite{bauschke2023introduction}.
\end{proof}

\subsection{Proof of Proposition \ref{linear_conv_prop}}
\begin{proof}
Existence and uniqueness of $\balpha$ follow from the strong convexity of $f(\bbeta)$. 
Because $\nabla g(\balpha \mid \balpha) = \nabla f(\balpha) = \mathbf{0}$, the 
$L$-smoothness of $g(\bbeta \mid \balpha)$ gives the quadratic upper bound
\begin{align}
    f(\bbeta) - f(\balpha) 
    &\le g(\bbeta \mid \balpha) - g(\balpha \mid \balpha) \notag \\
    &\le \nabla g(\balpha \mid \balpha)^\top (\bbeta - \balpha)
    + \frac{L}{2}\|\bbeta - \balpha\|_2^{2} \label{eq:prop2-pl-upper} \\
    &= \frac{L}{2}\|\bbeta - \balpha\|_2^{2}, \notag
\end{align}
which incidentally implies $\mu \le L$. In view of the strong convexity assumption, 
we have the lower bound
\begin{align}
    \|\nabla f(\bbeta)\|_2 \cdot \|\balpha - \bbeta\|_2
    &\ge -\nabla f(\bbeta)^\top (\balpha - \bbeta) \notag \\
    &\ge f(\balpha) - f(\bbeta) - \nabla f(\bbeta)^\top (\balpha - \bbeta)
    \label{eq:prop2-pl-lower} \\
    &\ge \frac{\mu}{2}\|\balpha - \bbeta\|_2^{2}. \notag
\end{align}
Therefore, $\|\nabla f(\bbeta)\|_2 \ge \frac{\mu}{2}\|\balpha - \bbeta\|_2$. 
Combining \eqref{eq:prop2-pl-upper} and \eqref{eq:prop2-pl-lower} yields the 
Polyak-\L{}ojasiewicz (PL) bound
\begin{align*}
    \|\nabla f(\bbeta)\|_2^2 &\ge \frac{\mu^{2}}{2L}\bigl[f(\bbeta) - f(\balpha)\bigr].
\end{align*}
We now turn to the MM iterates and take 
$\bbeta = \bbeta_{m} - \frac{1}{L}\nabla f(\bbeta_{m})$. 
The PL inequality implies
\begin{align*}
    f(\bbeta_{m+1}) - f(\bbeta_{m})
    &\le g(\bbeta_{m+1} \mid \bbeta_{m}) - g(\bbeta_{m} \mid \bbeta_{m}) \\
    &\le g(\bbeta \mid \bbeta_{m}) - g(\bbeta_{m} \mid \bbeta_{m}) \\
    &\le -\frac{1}{L}\nabla f(\bbeta_{m})^\top \nabla f(\bbeta_{m}) \\
    &\quad + \frac{1}{2L^2} \cdot L \|\nabla f(\bbeta_{m})\|_2^{2} \\
    &= -\frac{1}{2L}\|\nabla f(\bbeta_{m})\|_2^{2} \\
    &\le -\frac{\mu^2}{4L^2}\bigl[f(\bbeta_{m}) - f(\balpha)\bigr],
\end{align*}
Subtracting $f(\balpha)$ from both sides and rearranging gives
\begin{align*}
    f(\bbeta_{m+1}) - f(\balpha) 
    &\le \left(1 - \frac{\mu^{2}}{4L^{2}}\right)\bigl[f(\bbeta_{m}) - f(\balpha)\bigr].
\end{align*}
Iteration of this inequality yields the claimed linear convergence.
\end{proof}

\end{appendix}

\bibliographystyle{imsart-number}
\bibliography{LSQtactics}  
\end{document}